\documentclass[10pt,onecolumn]{IEEEtran}
\usepackage{subcaption}
\usepackage{algorithm}
\usepackage{algpseudocode}

\usepackage{pgf,tikz}
\usepackage{mathrsfs}
\usetikzlibrary{arrows,patterns}

\usepackage{bm}

\usepackage{amssymb}

\usepackage{amsmath}

\usepackage{amsthm}

\usepackage{mathtools}

\newtheorem{theorem}{Theorem}
\newtheorem{definition}[theorem]{Definition}
\newtheorem{lemma}[theorem]{Lemma}
\newtheorem{remark}[theorem]{Remark}

\newtheorem{corollary}[theorem]{Corollary}\newtheorem{construction}[theorem]{Construction}
\newtheorem{example}[theorem]{Example}
\newtheorem{proposition}[theorem]{Proposition}

\usepackage{hyperref}
\hypersetup{
	colorlinks,
	linkcolor={blue!100!black},
	citecolor={blue!100!black},
	urlcolor={blue!80!black}
}

\newcommand{\blue}[1]{\textcolor{black}{#1}}
\newcommand{\purple}[1]{\textcolor{black}{#1}}

\newcommand{\GC}{\textit{Gradient Coding}}

\newcommand{\bC}{\mathbb{C}}

\newcommand{\bE}{\mathbb{E}}
\newcommand{\bF}{\mathbb{F}}

\newcommand{\bK}{\mathbb{K}}

\newcommand{\bR}{\mathbb{R}}

\newcommand{\cA}{\mathcal{A}}
\newcommand{\cB}{\mathcal{B}}
\newcommand{\cC}{\mathcal{C}}

\newcommand{\cE}{\mathcal{E}}
\newcommand{\cF}{\mathcal{F}}

\newcommand{\cK}{\mathcal{K}}
\newcommand{\cL}{\mathcal{L}}

\newcommand{\cN}{\mathcal{N}}

\newcommand{\cP}{\mathcal{P}}
\newcommand{\cQ}{\mathcal{Q}}
\newcommand{\cR}{\mathcal{R}}
\newcommand{\cS}{\mathcal{S}}

\newcommand{\cV}{\mathcal{V}}
\newcommand{\cW}{\mathcal{W}}
\newcommand{\cX}{\mathcal{X}}

\newcommand{\boldA}{\textbf{A}}
\newcommand{\boldB}{\textbf{B}}
\newcommand{\boldC}{\textbf{C}}
\newcommand{\boldD}{\textbf{D}}

\newcommand{\boldG}{\textbf{G}}

\newcommand{\boldI}{\textbf{I}}

\newcommand{\boldN}{\textbf{N}}

\newcommand{\boldP}{\textbf{P}}

\newcommand{\boldU}{\textbf{U}}
\newcommand{\boldV}{\textbf{V}}

\newcommand{\bolda}{\bm{a}}
\newcommand{\boldb}{\bm{b}}
\newcommand{\boldc}{\bm{c}}

\newcommand{\boldf}{\bm{f}}
\newcommand{\boldg}{\bm{g}}

\newcommand{\bolds}{\bm{s}}
\newcommand{\boldt}{\bm{t}}
\newcommand{\boldu}{\bm{u}}
\newcommand{\boldv}{\bm{v}}
\newcommand{\boldw}{\bm{w}}
\newcommand{\boldx}{\bm{x}}
\newcommand{\boldy}{\bm{y}}
\newcommand{\boldz}{\bm{z}}

\DeclarePairedDelimiter{\floor}{\lfloor}{\rfloor}
\DeclarePairedDelimiter{\ceil}{\lceil}{\rceil}
\DeclareMathOperator{\supportOP}{supp}
\DeclareMathOperator{\diag}{diag}

\newcommand{\Span}[1]{{\left\langle {#1} \right\rangle}}
\newcommand{\zeronorm}[1]{\lVert #1 \rVert_0}
\newcommand{\support}[1]{\supportOP(#1)}
\newcommand{\norm}[1]{\lVert #1 \rVert_2}

\newcommand{\specnorm}[1]{\lVert #1 \rVert_{\text{spec}}}
\newcommand{\bigfloor}[1]{\bigl\lfloor #1 \bigr\rfloor}

\DeclareSymbolFont{bbold}{U}{bbold}{m}{n}
\DeclareSymbolFontAlphabet{\mathbbold}{bbold}
\newcommand{\1}{\mathbbold{1}}

\newif\ifSUPPLEMENTARY
\SUPPLEMENTARYtrue

\title{Gradient Coding from Cyclic MDS\\Codes and Expander Graphs}
\author{\textbf{Netanel Raviv}$^\star$, \textbf{Itzhak Tamo}$^\dagger$, \textbf{Rashish Tandon}$^\ddagger$, and \textbf{Alexandros G.~Dimakis}$^\circ$\\
\IEEEauthorblockA{\normalsize {$^\star$Department of Electrical Engineering, California Institute of Technology, Pasadena, CA, USA.\\
		$^\dagger$Department of Electrical Engineering--Systems, Tel-Aviv University, Israel.\\
		$^\ddagger$Apple, Seattle, WA, USA.\\
		$^\circ$Department of Electrical and Computer Engineering, The University of Texas at Austin, Austin, TX, USA.}
\thanks{Parts of this work were presented at the International Conference on Machine Learning (ICML), Stockholm, Sweden, 2018.}}}

\begin{document}
\maketitle
\thispagestyle{empty}

\begin{abstract}
Gradient coding is a technique for straggler mitigation in distributed learning. In this paper we design novel gradient codes using tools from classical coding theory, namely, cyclic MDS codes, which compare favorably with existing solutions, both in the applicable range of parameters and in the complexity of the involved algorithms. Second, we introduce an approximate variant of the gradient coding problem, in which we settle for approximate gradient computation instead of the exact one. This approach enables graceful degradation, i.e., the~$\ell_2$ error of the approximate gradient is a decreasing function of the number of stragglers. Our main result is that normalized adjacency matrices of expander graphs yield excellent approximate gradient codes, which enable significantly less computation compared to exact gradient coding, \blue{and guarantee faster convergence than trivial solutions under standard assumptions}. We experimentally test our approach on Amazon EC2, and show that the generalization error of approximate gradient coding is very close to the full gradient while requiring significantly less computation from the workers. 



\end{abstract}


\section{Introduction}\label{section:intro}
Data intensive machine learning tasks have become ubiquitous in many real-world applications, and with the increasing size of training data, distributed methods have gained increasing popularity. However, the performance of distributed methods (in synchronous settings) is strongly dictated by \textit{stragglers}, i.e., nodes that are slow to respond or unavailable. In this paper, we focus on coding theoretic (and graph theoretic) techniques for mitigating stragglers in distributed synchronous gradient descent.

A coding theoretic framework for straggler mitigation 
called \textit{gradient coding} was first introduced in~\cite{TandonLDK17}. It consists of a system with one master and~$n$ worker nodes (or servers), in which the data is partitioned into~$k$ parts, and one or more parts is assigned to each one of the workers. In turn, each worker computes the partial gradient on each of its assigned parts, linearly combines the results according to some predetermined vector of coefficients, and sends this linear combination back to the master node. Choosing the coefficients at each node judiciously, one can guarantee that the master node is capable of reconstructing the full gradient 
even if \textit{any} $s$ machines fail to perform their work. 
The \textit{storage overhead} of the system, which is denoted by~$d$, refers to the amount of redundant computations, or alternatively, to the number of data parts that are sent to each node (see example in Fig.~\ref{figure:example}).

The importance of straggler mitigation was demonstrated in a series of recent studies (e.g.,~\cite{LiMu} and~\cite{MultiTask}). In particular, it was demonstrated in~\cite{TandonLDK17} that stragglers may run up to $\times 5$ slower than the typical worker ($\times 8$ in~\cite{MultiTask}) on 
Amazon EC2, especially for the cheaper virtual machines; 
such erratic behavior is unpredictable and can significantly 
delay training. One can, of course, use more expensive instances but the goal here is to use coding theoretic methods to provide reliability 
out of cheap unreliable workers, overall reducing the cost of training. 

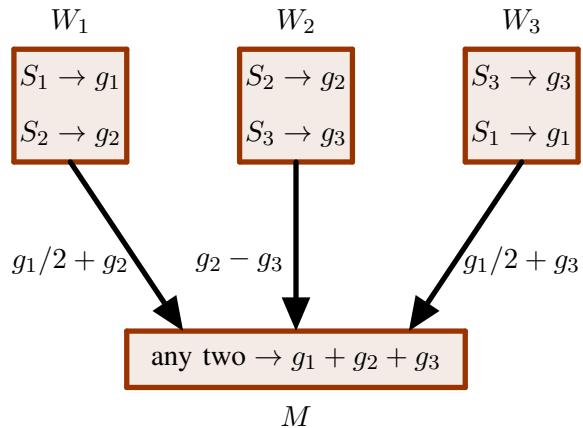
\begin{figure}
    \centering
    \definecolor{zzttqq}{rgb}{0.6,0.2,0}
\begin{tikzpicture}[scale=0.75,line cap=round,line join=round,>=triangle 45,x=1cm,y=1cm]\clip(0,3) rectangle (12,11);
\fill[line width=2pt,color=zzttqq,fill=zzttqq,fill opacity=0.10000000149011612] (3,5) -- (9,5) -- (9,4) -- (3,4) -- cycle;
\fill[line width=2pt,color=zzttqq,fill=zzttqq,fill opacity=0.10000000149011612] (5,8) -- (7,8) -- (7,10) -- (5,10) -- cycle;
\fill[line width=2pt,color=zzttqq,fill=zzttqq,fill opacity=0.10000000149011612] (3,8) -- (3,10) -- (1,10) -- (1,8) -- cycle;
\fill[line width=2pt,color=zzttqq,fill=zzttqq,fill opacity=0.10000000149011612] (9,8) -- (9,10) -- (11,10) -- (11,8) -- cycle;
\draw [line width=2pt,color=zzttqq] (3,5)-- (9,5);
\draw [line width=2pt,color=zzttqq] (9,5)-- (9,4);
\draw [line width=2pt,color=zzttqq] (9,4)-- (3,4);
\draw [line width=2pt,color=zzttqq] (3,4)-- (3,5);
\draw [line width=2pt,color=zzttqq] (5,8)-- (7,8);
\draw [line width=2pt,color=zzttqq] (7,8)-- (7,10);
\draw [line width=2pt,color=zzttqq] (7,10)-- (5,10);
\draw [line width=2pt,color=zzttqq] (5,10)-- (5,8);
\draw [line width=2pt,color=zzttqq] (3,8)-- (3,10);
\draw [line width=2pt,color=zzttqq] (3,10)-- (1,10);
\draw [line width=2pt,color=zzttqq] (1,10)-- (1,8);
\draw [line width=2pt,color=zzttqq] (1,8)-- (3,8);
\draw [line width=2pt,color=zzttqq] (9,8)-- (9,10);
\draw [line width=2pt,color=zzttqq] (9,10)-- (11,10);
\draw [line width=2pt,color=zzttqq] (11,10)-- (11,8);
\draw [line width=2pt,color=zzttqq] (11,8)-- (9,8);
\draw [->,line width=2pt] (2,8) -- (4,5);
\draw [->,line width=2pt] (6,8) -- (6,5);
\draw [->,line width=2pt] (10,8) -- (8,5);
\draw (2,9.5) node[anchor=center] {$\cS_1 \rightarrow \boldg_1$};
\draw (2,8.5) node[anchor=center] {$\cS_2 \rightarrow \boldg_2$};
\draw (10,9.5) node[anchor=center] {$\cS_3 \rightarrow \boldg_3$};
\draw (10,8.5) node[anchor=center] {$\cS_1 \rightarrow \boldg_1$};
\draw (6,9.5) node[anchor=center] {$\cS_2 \rightarrow \boldg_2$};
\draw (6,8.5) node[anchor=center] {$\cS_3 \rightarrow \boldg_3$};
\draw (2,6.25) node[anchor=center] {$\boldg_1/2+\boldg_2$};
\draw (5,6.25) node[anchor=center] {$\boldg_2-\boldg_3$};
\draw (10,6.25) node[anchor=center] {$\boldg_1/2+\boldg_3$};
\draw (6,3.5) node[anchor=center] {$M$};
\draw (6,4.5) node[anchor=center] {$\mbox{any two}
\rightarrow \boldg_1+\boldg_2+\boldg_3$};
\draw (2,10.5) node[anchor=center] {$W_1$};
\draw (6,10.5) node[anchor=center] {$W_2$};
\draw (10,10.5) node[anchor=center] {$W_3$};
\end{tikzpicture}
    \caption{Gradient coding for~$n=3$, $k=3$, $d=2$, and~${s=1}$~\cite{TandonLDK17}. Each worker node~$W_i$ obtains two parts~$\cS_{i_1},\cS_{i_2}$ of the data set~$\cS=\cS_1\cup \cS_2\cup \cS_3$, computes the partial gradients~$\boldg_{i_1},\boldg_{i_2}$, and sends their linear combination back to the master node~$M$. By choosing the coefficients judiciously, the master node~$M$ can compute the full gradient from any two responses, providing robustness against any one straggler.}
    \label{figure:example}
\end{figure}

\blue{By and large, the purpose of gradient coding is to enable the master node to compute the exact gradient out of the responses of \textit{any} $n-s$ non-straggling nodes.} The work of~\cite{TandonLDK17} established the fundamental bound~$d\ge s+1$, provided a deterministic construction which achieves it with equality when~$s+1\vert n$, and a randomized one which applies to all~$s$ and~$n$. Subsequently, deterministic constructions were also obtained by~\cite{ShortDot} and~\cite{Wael}. These works have focused on the scenario where~$s$ is known prior to the construction of the system. Furthermore, the exact computation of the full gradient is guaranteed if the number of stragglers is at most~$s$, but no error bound is guaranteed if this number exceeds~$s$.

The contribution of this work is twofold. For the computation of the exact gradient we employ tools from classic coding theory, namely, cyclic MDS codes, in order to obtain a deterministic construction which compares favourably with existing solutions; both in the applicable range of parameters, and in the complexity of the involved algorithms. Some of these gains are a direct application of well known properties of these codes. 

Second, we introduce an \textit{approximate} variant of the gradient coding problem. In this variant, the requirement for exact computation of the full gradient is traded by an approximate one, where the~$\ell_2$-deviation of the given solution \blue{from the full gradient decreases as the number of stragglers decreases. }
Note that by this approach, the parameter~$s$ is \textit{not} a part of the system construction, and the system can provide an approximate solution for any~$s<n$, whose quality deteriorates gracefully as~$s$ increases. In the suggested solution, the coefficients at the worker nodes are based on an important family of graphs called \textit{expanders}. In particular, it is shown that setting these coefficients according to a normalized adjacency matrix of an expander graph, a strong bound on the error term of the resulting solution is obtained. Moreover, this approach enables to break the aforementioned barrier~$d\ge s+1$, which is a substantial obstacle in gradient coding, and allows the master node to decode using a very simple algorithm.

This paper is organized as follows. Related work regarding gradient coding (and coded computation in general) is listed in Section~\ref{section:relatedWork}. A framework which encapsulates all the results in this paper is given in Section~\ref{section:framework}. Necessary mathematical notions from coding theory and graph theory are given in Section~\ref{section:notions}. The former notions are used to obtain an algorithm for exact gradient computation in Section~\ref{section:exact}, and the latter ones are used to obtain an algorithm for the approximate gradient in Section~\ref{section:approximate}. Experimental results are given in Section~\ref{section:experimentalResults}. 

\section{Related Work}\label{section:relatedWork}
The work of Lee et al.~\cite{LeeSpeedingUp} initiated the use of coding theoretic methods for mitigating stragglers in large-scale learning. This work is focused on linear regression and therefore can exploit more structure compared to the general gradient coding problem that we study here. The work by Li et al.~\cite{LiGlobecom}, investigates a generalized view of the coding ideas in~\cite{LeeSpeedingUp}, showing that their solution is a single operating point in a general scheme of trading off latency of computation to the load of communication. 

Further closely related work has shown how coding can be used for distributed MapReduce, as well as a similar communication and computation tradeoff~\cite{CodedMapReduce,li2018fundamental}.  We also mention the work of~\cite{Karakus} which addresses straggler mitigation in linear regression by using a different approach, that is not mutually exclusive with gradient coding. In their work, the data is coded rather than replicated, and the nodes perform their computation on coded data.

The work by~\cite{ShortDot} generalizes previous work for linear models~\cite{LeeSpeedingUp} but can also be applied to general models to yield explicit gradient coding constructions. Our results regarding the exact gradient are closely related to the work by~\cite{Wael,Wael2} which was obtained independently from our work. In~\cite{Wael}, similar coding theoretic tools were employed in a fundamentally different fashion. Both~\cite{Wael} and~\cite{ShortDot} are comparable in parameters to the randomized construction of~\cite{TandonLDK17} and are outperformed by us in a wide range of parameters. A detailed comparison of the theoretical asymptotic behaviour is given in the sequel.

\begin{remark}
	None of the aforementioned works studies approximate gradient computations. However, we note that subsequent to this work, two unpublished manuscripts~\cite{Dimitris,li2017near} study a similar approximation setting and obtain related results albeit using randomized as opposed to deterministic approaches. Furthermore, the exact setting was also discussed subsequent to this work in~\cite{ye2018communication} and~\cite{yu2018lagrange}. In~\cite{ye2018communication} it was shown that network communication can be reduced by increasing the replication factor, and respective bounds were given. The work of~\cite{yu2018lagrange} discussed coded polynomial computation with low overhead, and applies to gradient coding whenever the gradient at hand is a polynomial.
\end{remark}

\section{Framework}\label{section:framework}
This section provides a unified framework which accommodates straggler mitigation in both the exact and approximate gradient computations which follow. \blue{By and large, we use lowercase letters $a,b,\ldots$ to refer to scalars and functions, bold lowercase letter~$\bolda,\boldb,\ldots$ to refer to vectors, uppercase letters~$A,B,\ldots$ to refer to workers (or nodes), uppercase bold letters~$\boldA,\boldB$ to refer to matrices, and uppercase calligraphic letters~$\cA,\cB,\ldots$ to refer to sets and codes. E.g., the rows of a matrix~$\boldB$ are denoted by~$\boldb_i$, and its entries are denoted by~$b_{i,j}$. Unless otherwise stated, all vectors are \textit{row} vectors, and for any field~$\bF$ we adopt the common notation~$\bF^n$ to refer to~$\bF^{1\times n}$. We also employ the standard notation~$[n]\triangleq \{1,2,\ldots,n\}$ for an integer~$n$}.

We begin with a brief introduction to machine learning, and the reader is referred to~\cite{shalev2014understanding} for further reading (in particular, to \textit{gradient descent}~\cite[Sec.~14.1]{shalev2014understanding} and \textit{stochastic gradient descent}~\cite[Sec.~14.3]{shalev2014understanding}). \blue{Broadly speaking, the general purpose of machine learning is to find a hypothesis from a given hypotheses class, which best approximates an unknown target function by observing a \textit{training set} $\cS=\{ \bolds_i \}_{i=1}^m$ that is labeled by that target function. Assuming that the hypotheses class is parametrized by real vector~$\boldw\in\bR^p$, one defines a \textit{loss function}~$\ell(\boldw,\bolds)$, which penalizes a given hypothesis~$\boldw$ for erring on a data point~$\bolds$. Then, one wishes to find the~$\boldw$ which minimizes the \textit{empirical risk} $L_\cS(\boldw)\triangleq\frac{1}{m}\sum_{i=1}^m \ell(\boldw,\bolds_i)$ by using analytic methods. That is, one starts from an arbitrary point~$\boldw_0$, and iteratively computes the \textit{gradient}~$\nabla L_\cS(\boldw)=(\frac{\partial}{\partial w_i}L_\cS(\boldw))_{i=1}^n$ at the current point, and moves away from the current point in the opposite direction of the gradient, until convergence is achieved. In the \textit{gradient descent} algorithm the exact gradient is computed at every iteration, and in the \textit{stochastic gradient descent} algorithm it is taken from a distribution whose expected value is the exact gradient.}

In order to distribute the execution of gradient descent from a master node~$M$ to~$n$ worker nodes~$\{W_j\}_{j=1}^n$ (\blue{Algorithm~\ref{algorithm:distributedSGD}}), the training set~$\cS$ is partitioned to~$n$ disjoint subsets~$\{ \cS_i \}_{i=1}^n$ of size\footnote{For simplicity, assume that~$m\vert n$. The given scheme could be easily adapted to the case~$m\nmid n$. Further, the assumption that the number of partitions equals to the number of nodes is a mere convenience, and all subsequent schemes can be adapted to the case where the number of partitions is at most the number of nodes.} $\frac{m}{n}$ each. These subsets are distributed among~$\{ W_j \}_{i=1}^n$, \blue{at most~$d$ subsets at each worker for some parameter~$d$}. Then, every node computes the partial gradients~$\nabla L_{\cS_i}(\boldw)$ on the~$\cS_i$'s which it obtained, \blue{where $L_{\cS_i}(\boldw)\triangleq\frac{n}{m}\sum_{\bolds\in \cS_i}\ell(\boldw,\bolds)$}. The algorithm operates in iterations, where in iteration~$r\in[t]$ every node evaluates its gradients in the current model~$\boldw^{(t)}$, and sends to~$M$ some linear combination of them. After obtaining responses from at least~\blue{$n-s_r$ workers, where~$s_r$ is the number of stragglers in iteration~$r$}, $M$~aggregates them to form the gradient~$\nabla L_\cS(\boldw^{(t)})$ of the overall empirical risk at~$\boldw^{(t)}$. In the exact setting the value of~$s_r$ will \blue{be some fixed~$s$} for every~$r$, whereas in the approximate setting this value is at the discretion of the master, in correspondence with the required approximation error.
	
To support mitigation of stragglers in this setting, the following notions are introduced. Let~$\boldB\in\bC^{n\times n}$ be a matrix whose~$i$'th row~$\boldb_i$ contains the coefficients of the linear combination $\sum_{j=1}^{n}b_{i,j}\cdot\nabla L_{\cS_j}(\boldw^{(t)})$ that is sent to~$M$ by~$W_i$. Note that the support~$\support{\boldb_i}$ contains the indices of the sets~$\cS_j$ that are to be sent to~$W_i$ by~$M$. Given a set of non-stragglers~$\cK\in\cP(n)$, where~$\cP(n)$ is the set of all nonempty subsets of~$[n]$, a function~$a:\cP(n)\to \bC^n$ provides~$M$ with a vector by which the results from~$\{ W_i \}_{i\in \cK}$ are to be linearly combined to obtain \blue{a vector~$\boldv_r\in\bR^p$. This vector is either the true gradient, in which case Algorithm~\ref{algorithm:distributedSGD} is gradient descent, or an approximation of it (whose expectation is the true gradient), in which case Algorithm~\ref{algorithm:distributedSGD} is stochastic gradient descent}. We also require that~$\support{a(\cK)}\subseteq \cK$ for all~$\cK\in\cP(n)$. In most of the subsequent constructions, the matrix~$\boldB$ and the function~$a$ will be defined over~$\bR$ rather than over~$\bC$. 

\begin{algorithm}[tb]
	\caption{-- Gradient Coding.}\label{algorithm:distributedSGD}
	\begin{algorithmic}
	\State \textbf{Input:} Data $\cS=\{\boldz_i=(\boldx_i,y_i)\}_{i=1}^m$, number of iterations~$t>0$, learning rates~$\{\eta\}_{r=1}^t>0$, straggler tolerance parameters~$\{s_r\}_{r=1}^t$, a matrix~$\boldB\in\bC^{n\times n}$, and a function~$a:\cP(n)\to \bC^n$.
	\State Initialize~$\boldw^{(1)}=(0,\ldots,0)$.
	\State Partition~$\cS=\cup_{i=1}^n \cS_i$ and send~$\{\cS_j : j\in\support{\boldb_i}\}$ to~$W_i$ for every~$i\in[n]$.
	\For{$r=1$ \textbf{to} $t$}
		\State $M$ broadcasts~$\boldw^{(r)}$ to all nodes.
		\State Each~$W_j$ sends $\frac{1}{n}\sum_{i\in\support{\boldb_j}}b_{j,i}\cdot \nabla L_{\cS_i}(\boldw^{(r)})$ to~$M$.\label{line:eachWorker}
		\State $M$ waits until at least~$n-s_r$ nodes responded.
		\State $M$ computes~$\boldv_r=a(\cK_r)\cdot \boldC$, where the~$i$'th of~$\boldC$ \blue{is~$\frac{1}{n}$ times} the response from~$W_i$, if it responded, and~$0$ otherwise, and $\cK_r$ is the set of non-stragglers in the current iteration~$r$.
		\State $M$ updates $\boldw^{(r+1)}\triangleq \boldw^{(r)}-\eta_r \boldv_r$
	\EndFor
	\State
	\textbf{Return}~$\frac{1}{t}\sum_{r=1}^t\boldw^{(r+1)}$. 
	\end{algorithmic}
\end{algorithm}

	Different constructions of the matrix~$\boldB$ and the function~$a$ in Algorithm~\ref{algorithm:distributedSGD} enable to compute the gradient either exactly (which requires the storage overhead~$d$ to be at least~$s_r+1$ for all~$r\in[t]$) or \textit{approximately}. In what follows, the respective requirements and guarantees from~$a$ and~$\boldB$ are discussed. In the following definitions, for an integer~$r$ let~$\1_r$ be the vector of~$r$ ones, where the subscript is omitted if clear from context, and for~$\cK\subseteq [n]$ let~$\cK^c\triangleq [n]\setminus \cK$.
	
	\begin{definition}\label{definition:ECcondition}
		A matrix~$\boldB\in\bC^{n\times n}$ and a function~$a:\cP(n)\to \bC^n$ satisfy the Exact Computation (EC) condition if for all~$\cK\subseteq[n]$ such that~$|\cK|\ge \max _{r\in[t]}s_r$, we have~$a(\cK)\cdot \boldB=\1$. 
	\end{definition}

	\begin{definition}\label{definition:ACcondition}
		For a non-decreasing function~$\epsilon:[n-1]\to\bR_{\ge 0}$ such that~$\epsilon(0)=0$, $a$ and~$\boldB$ satisfy the~$\epsilon$-Approximate Computation ($\epsilon$-AC) condition, if for all~$\cK\in\cP(n)$, we have~$d_2(a(\cK)\boldB,\1)\le \epsilon(|\cK^c|)$ (where~$d_2$ is Euclidean distance).
	\end{definition}
	
	Notice that the error term~$\epsilon$ in Definition~\ref{definition:ACcondition} is a function of the number of stragglers since it is not expected to decrease if more stragglers are present. The conditions which are given in Definition~\ref{definition:ECcondition} and Definition~\ref{definition:ACcondition} guarantee the exact and approximate computation by the following lemmas. In the upcoming proofs, let~$\boldN(\boldw)$ be the matrix
	\begin{align}\label{equation:N(w)}
	\boldN(\boldw)\triangleq \dfrac{1}{n}\cdot
	\begin{pmatrix}
	\nabla L_{\cS_1}(\boldw)\\
	\nabla L_{\cS_2}(\boldw)\\
	\vdots\\
	\nabla L_{\cS_n}(\boldw)\\
	\end{pmatrix}.
	\end{align}
	
	\begin{lemma}\label{lemma:ECguarantee}
		If~$a$ and~$\boldB$ satisfy the EC condition, then for all~$r\in[t]$ we have~$\boldv_r=\nabla L_\cS(\boldw^{(r)})$.
	\end{lemma}

	\begin{proof} 
		For a given~$r\in[t]$, let~$\boldB'$ be the matrix whose~$i$'th row~$\boldb_i'$ equals~$\boldb_i$ if~$i\in \cK_r$, and zero otherwise. By the definition of~$\boldC$ \blue{in Algorithm~\ref{algorithm:distributedSGD}} it follows that~$\boldC=\boldB'\cdot \boldN(w^{(r)})$, and since~$\support{a(\cK_r)}\subseteq \cK_r$ it follows that~$a(\cK_r)\boldB'=a(\cK_r)\boldB$. Therefore, we have
	\begin{align*}
	\boldv_r&=a(\cK_r)\cdot \boldC=a(\cK_r)\cdot \boldB'\cdot \boldN(\boldw^{(r)})\\
	&=a(\cK_r)\cdot \boldB\cdot \boldN(\boldw^{(r)})=\1\cdot \boldN(\boldw^{(r)})\\
	&=\dfrac{1}{n}\sum_{i=1}^n\nabla L_{\cS_i}(\boldw)=\frac{1}{n}\sum_{i=1}^n\frac{1}{m/n}\sum_{\boldz\in \cS_i}\nabla\ell(\boldw^{(r)},\boldz)\\
	&=\frac{1}{m}\sum_{\boldz\in \cS}\nabla\ell(\boldw^{(r)},\boldz)=\nabla L_\cS(\boldw^{(r)})\qedhere.
	\end{align*}
\end{proof}
	The next lemma bounds the deviance of~$\boldv_r$ from the gradient of the empirical risk at the current model~$\boldw^{(r)}$ by using the function~$\epsilon$ and the \textit{spectral norm}~$\specnorm{\cdot}$ of $\boldN(\boldw)$. Recall that for a matrix~$\boldP$ the spectral norm is defined as~$\specnorm{\boldP}\triangleq \max_{\boldx: \norm{\boldx}=1}\norm{\boldP \boldx}$.
	
	\begin{lemma}\label{lemma:epsACguarantee}
		For a function~$\epsilon$ as above,	if~$a$ and~$\boldB$ satisfy the~$\epsilon$-AC condition, then~$d_2(\boldv_r,\nabla L_\cS(\boldw^{(r)}))\le \epsilon(|\cK_r^c|)\cdot \specnorm{\boldN(\boldw^{(r)})}$.
	\end{lemma}
	\begin{proof}
	As in the proof of Lemma~\ref{lemma:ECguarantee}, we have that
	\begin{align*}
	d_2&(\boldv_r,\nabla L_\cS(\boldw^{(r)}))=\\
	&=d_2(a(\cK_r)\cdot \boldB'\cdot \boldN(\boldw^{(r)}),\1\cdot \boldN(\boldw^{(r)}))\\
	&=d_2(a(\cK_r)\cdot \boldB\cdot \boldN(\boldw^{(r)}),\1\cdot \boldN(\boldw^{(r)})\\
	&=\norm{(a(\cK_r)\boldB-\1)\cdot \boldN(\boldw^{(r)})}\\
	&\le d_2(a(\cK_r)\boldB,\1)\cdot\specnorm{\boldN(\boldw^{(r)})}\\
	&\le \epsilon(|\cK_r^c|)\cdot \specnorm{\boldN(\boldw^{(r)})}.\qedhere
	\end{align*}
\end{proof}
	Due to Lemma~\ref{lemma:ECguarantee} and Lemma~\ref{lemma:epsACguarantee}, in the remainder of this paper we focus on constructing~$a$ and~$\boldB$ that satisfy either the EC condition (Section~\ref{section:exact}) or the~$\epsilon$-AC condition (Section~\ref{section:approximate}). 
	
		\begin{remark}\label{remark:k<n}
		In some settings~\cite{TandonLDK17}, it is convenient to partition the data set~$\cS$ to~$\{\cS_i\}_{i=1}^k$, where~$k<n$. Notice that the definitions of~$a$ and~$\boldB$ above extend verbatim to this case as well. If~$a$ and~$\boldB$ satisfy the~EC condition, we have that~$a(\cK)\boldB=\1_n$ for every large enough~$\cK\subseteq [n]$. Hence, by omitting any~$n-k$ columns of~$\boldB$ to form a matrix~$\hat{\boldB}$, we have that~$a(\cK)\hat{\boldB}=\1_k$, and hence a scheme for any partition of~$\cS$ to~$k$ parts emerges instantly. This new scheme~$(a,\hat{\boldB})$ is resilient to an identical number of stragglers~$s$ and has lesser or equal storage overhead than~$(a,\boldB)$. Similarly, if~$a$ and~$\boldB$ satisfy the~$\epsilon$-AC condition for some~$\epsilon$, then the scheme~$(a,\hat{\boldB})$ has lesser or equal storage overhead, and an identical error function~$\epsilon$, since~$d_2(a(\cK)\hat{\boldB},\1_k)\le d_2(a(\cK)\boldB,\1_n)\le \epsilon(|\cK^c|)$ for any~$\cK\in\cP(n)$.
	\end{remark}
	
\section{Mathematical Notions}\label{section:notions}
This section provides a brief overview on the mathematical notions that are essential for the suggested schemes. The exact computation (Sec.~\ref{section:exact}) requires notions from coding theory, and the approximate one (Sec.~\ref{section:approximate}) requires notions from graph theory. The coding theoretic material in this section is taken from~\cite{Ronny}, which focuses on finite fields, and yet the given results extend verbatim to the real or complex case (see also~\cite{RealCodes}, Sec.~8.4).

	
For a field~$\bF\in\{\bR,\bC\}$ an~$[n,\kappa]$ (linear) code~$\cC$ over~$\bF$ is a subspace of dimension~$\kappa$ of~$\bF^n$. The minimum distance~$\delta$ of~$\cC$ is~$\min\{ d_H(\boldx,\boldy):\boldx,\boldy\in \cC,~\boldx\ne \boldy \}$, where~$d_H$ denotes the \textit{Hamming distance}~$d_H(\boldx,\boldy)\triangleq\left| \{i\vert x_i\ne y_i\} \right|$. Since the code is a linear subspace, it follows that the minimum distance of a code is equal to the minimum \textit{Hamming weight}~$w_H(\boldx)\triangleq\zeronorm{\boldx}=|\support{\boldx}|$ among the nonzero codewords in~$\cC$. The well-known \textit{Singleton} bound states that~$\delta\le n-\kappa+1$, and codes which attain this bound with equality are called \textit{Maximum Distance Separable} (MDS) codes. 

\begin{definition}\cite[Sec.~8]{Ronny}\label{definition:cyclicCode}
	A code~$\cC$ is called cyclic if the cyclic shift of every codeword is also a codeword, namely,
	\begin{align*}
		(c_1,c_2\ldots,c_n)\in \cC \Rightarrow (c_n,c_1,\ldots,c_{n-1})\in \cC.
	\end{align*}
\end{definition}
The \textit{dual code} of~$\cC$ is the subspace~$\cC^\bot\triangleq\{\boldy\in \bF^n\vert \boldy\cdot \boldc^\top=0\mbox{ for all }\boldc\in \cC \}$. Several well-known and easy to prove properties of MDS codes are used throughout this paper. 
	
	\begin{lemma}\label{lemma:MDSproperties}
		If~$\cC\subseteq \bF^n$ is an~$[n,\kappa]$ MDS code, then
		\begin{itemize}
			\item[A1.]  $\cC^\bot$ is an~$[n,n-\kappa]$ MDS code, and hence its minimum Hamming weight is~$\kappa+1$.
			\item[A2.] For any subset~$\cK\subseteq [n]$ of size~$n-\kappa+1$ there exists a codeword in~$\cC$ whose support (i.e., the set of nonzero indices) is~$\cK$.
			\item[A3.] The reverse code~$\cC^R \triangleq \{(c_n,\ldots,c_1)\vert (c_1,\ldots,c_n)\in \cC \}$ is an~$[n,\kappa]$ MDS code.
		\end{itemize}
	\end{lemma}

\begin{proof}
	For A1 see \cite[Prob.~4.1]{Ronny}. The proof of A3 is trivial since permuting the coordinates of a code does not alter its minimum distance. For A2, let~$\boldG\in\bF^{\kappa\times n}$ be the \textit{generator matrix} of~$\cC$, i.e., a matrix whose rows are a basis to~$\cC$. By restricting~$\boldG$ to the columns indexed by~$\cK^c$ we get a~$\kappa\times (\kappa-1)$ matrix, which has a nonzero vector~$\boldy$ in its left kernel, and hence~$\boldy\boldG$ is a codeword in~$\cC$ which is zero in the entries that are indexed by~$\cK^c$. Since the minimum distance of~$\cC$ is~$n-\kappa+1$, it follows that~$\boldy\boldG$ has nonzero values in entries that are indexed by~$\cK$, and the claim follows.
\end{proof}

Two common families of codes are used in the sequel---Reed-Solomon (RS) codes and\allowbreak~Bose-Chaudhuri-Hocquen\-ghem (BCH) codes. An~$[n,\kappa]$ RS code~$\cC$ is defined by a set of~$n$ distinct \textit{evaluation points}~$\alpha_0,\ldots,\alpha_{n-1}\in\bF$ as
$$\cC=\{ (f(\alpha_0),f(\alpha_1),\ldots,f(\alpha_{n-1})) : f\in \bF^{<\kappa}[x] \},$$ where~$\bF^{<\kappa}[x]$ is the set of polynomials of degree less than~$\kappa$ and coefficients from~$\bF$ in the variable~$x$. Alternatively, RS codes can be defined as~$\{ \boldy \boldV\vert \boldy\in\bF^\kappa \}$, where~$\boldV\in\bF^{\kappa\times n}$ is a Vandermonde matrix on~$\{ \alpha_i \}_{i=0}^{n-1}$, i.e., $v_{i,j}=\alpha_{j-1}^{i-1}$ for every~$(i,j)\in[\kappa]\times [n]$.
It is widely known that RS codes are MDS codes, and in some cases, they are also cyclic.

In contrast with RS codes, where every codeword is a vector of \textit{evaluations} of a polynomial, a codeword in a BCH code is a vector of \textit{coefficients} of a polynomial; that is, a codeword $\boldc=(c_0,c_1,\ldots,c_{n-1})$ is identified by $c(x)\triangleq c_0+c_1x+\ldots+c_{n-1}x^{n-1}$. For a field~$\bK$ that contains~$\bF$ and a set~$\cR\subseteq \bK$, a BCH code~$\cC$ is defined as~$\cC=\{ c\in\bF^n  \vert c(r)=0 \mbox{ for all }r\in \cR \}$. The set~$\cR$ is called \textit{the roots of~$\cC$}, or alternatively, $\cC$ is said to be a BCH code on~$\cR$ over~$\bF$.
For example, a set of complex numbers~$\cR\subseteq \bC$ defines a BCH code on~$\cR$ over~$\bR$, which is the set of real vectors whose corresponding polynomials vanish on~$\cR$.

\begin{lemma}\cite{RealCodes,Ronny}\label{lemma:BCHcyclic}
	If all elements of~$\cR\subseteq \bC$ are roots of unity of order~$n$, then the BCH code~$\cC$ on~$\cR$ over~$\bR$ is cyclic.
\end{lemma}

\begin{proof}
	If~$c(x)$ is a codeword in~$\cC$, then its cyclic shift is given by~$\tilde{c}(x)\triangleq c(x)\cdot x \bmod (x^n-1)=x\cdot c(x)-c_{n-1}\cdot (x^n-1)$. Since every~$r\in \cR$ is a root of unity of order~$n$, it follows that
	\begin{align*}
	\tilde{c}(r)=r \cdot c(r)-c_{n-1}\cdot(r^n-1)=r\cdot c(r)=0,
	\end{align*}
	and hence~$\tilde{c}$ is a codeword in~$\cC$.
\end{proof}

Further, the structure of~$R$ may also imply a lower bound on the distance of~$C$.

\begin{theorem}(The BCH bound)~\cite{RealCodes,Ronny}\label{theorem:BCHbound}
	If~$\cR$ contains a subset of~$\ell$ consecutive powers of a primitive root of unity (i.e., a subset of the form~$\omega^{b},\omega^{b+1},\ldots,\omega^{b+\ell-1}$, where~$\omega$ is a  primitive $n$'th root of unity), then the minimum distance of~$\cC$ is at least~$\ell+1$.
\end{theorem}

In the remainder of this section, a brief overview on expander graphs is given. The interested reader is referred to~\cite{Expanders1} for further details. Let~$G=(\cV,\cE)$ be a~$d$-regular, undirected, and connected graph on~$n$ nodes. Let~$\boldA_G\in\bR^{n\times n}$ be the \textit{adjacency matrix} of~$G$, i.e., $(\boldA_G)_{i,j}=1$ if $\{i,j \}\in \cE$, and~$0$ otherwise. Since $\boldA_G$ is a real symmetric matrix, it follows that it has~$n$ real eigenvalues~$\lambda_1\ge\lambda_2\ge \ldots \ge \lambda_n$, and denote~$\lambda\triangleq \max\{|\lambda_2|,|\lambda_n| \}$. It is widely known~\cite{Expanders1} that $\lambda_1=d$, and that $\lambda_n\ge -d$, where equality holds if and only if~$G$ is bipartite. Further, it also follows from~$\boldA_G$ being real and symmetric that it has a basis of orthogonal real eigenvectors $\boldv_1=\1,\boldv_2,\ldots,\boldv_n$, and w.l.o.g assume that $\norm{\boldv_i}=1$ for every~$i\ge 2$. The parameters~$\lambda$ and~$d$ are related by the celebrated Alon-Boppana Theorem.
	
\begin{theorem}\label{theorem:lambdaLowerBound}\cite{Expanders1}
		An infinite sequence of~$d$ regular graphs on~$n$ vertices satisfies that $\lambda\ge 2\sqrt{d-1}-o_n(1)$, where~$o_n(1)$ is an expression which tends to zero as~$n$ tends to infinity.
	\end{theorem}
	
	Constant degree regular graphs (i.e., families of graphs with fixed degree~$d$ that does not depend on~$n$) for which~$\lambda$ is small in comparison with~$d$ are largely referred to as \textit{expanders}. In particular, graphs which attain the above bound asymptotically (i.e.,~$\lambda\le 2\sqrt{d-1}$) are called \textit{Ramanujan graphs}, and several efficient constructions are known~\cite{RamanujanGraphs,cohen2016ramanujan}. \blue{Since explicit constructions of Ramanujan graphs are often rather intricate, one may resort to choosing a random regular graph and verify its expansion by computing the respective eigenvalues. This process that is known to produce a good expander with high probability~\cite[Theorem~7.10]{Expanders1}, and is used in our experiments.}

\section{Exact Gradient Coding from Cyclic MDS Codes}\label{section:exact}

	
For a given~$n$ and~$s$, let~$\cC$ be a cyclic $[n,n-s]$ MDS code over~$\bF$ that contains~$\1$ (explicit constructions of such codes are given in the sequel). According to Lemma~\ref{lemma:MDSproperties}, there exists a codeword $\boldc_1\in \cC$ whose support is~$\{1,\ldots,s+1\}$. Let~$\boldc_2,\ldots,\boldc_n$ be all cyclic shifts of~$\boldc_1$, which lie in~$\cC$ by its cyclic property. Finally, let~$\boldB$ be the~$n\times n$ matrix whose \textit{columns} are $\boldc_1,\ldots,\boldc_n$, i.e., $\boldB\triangleq (\boldc_1^\top,\boldc_2^\top,\ldots,\boldc_n^\top)$. The following lemma provides some important properties of~$\boldB$.
	
	\begin{lemma} \label{lemma:matrixB}
		The matrix~$\boldB$ satisfies the following properties.
		\begin{itemize}
			\item [B1.] $\zeronorm{\boldb}=s+1$ for every row~$\boldb$ of~$\boldB$.
			\item [B2.] Every row of~$\boldB$ is a codeword in~$\cC^R$.
			\item [B3.] The column span of~$\boldB$ is the code~$\cC$.
			\item [B4.] Every set of~$n-s$ rows of~$\boldB$ are linearly independent over~$\bF$.
		\end{itemize}
	\end{lemma}

\begin{proof}
	To prove~B1 and~B2, observe that~$\boldB$ is of the following form, where~$\boldc_1\triangleq (\beta_1,\ldots,\beta_{s+1},0,\ldots,0)$.
	\begin{align*}
		\begin{pmatrix}
			\beta_1      & 0           & \cdots & 0            & \beta_{s+1} & \beta_{s}   & \ldots & \beta_2\\
			\beta_2      & \beta_1     & 0 		& \cdots       & 0 	         & \beta_{s+1} & \ldots & \beta_3\\
			\vdots       & \vdots      & \ddots & \ddots       & \vdots      & \ddots      & \ddots & \vdots  \\
			\beta_{s}    & \beta_{s-1} & \cdots & \beta_1            & 0      & \cdots       & 0      & \beta_{s+1}\\
			\beta_{s+1}  & \beta_s     & \cdots & \beta_2      & \beta_1           & 0      & \cdots       & 0\\
			0            & \beta_{s+1} & \cdots & \beta_3      & \beta_2     & \beta_1           & \cdots & 0\\
			\vdots       & \vdots      & \ddots & \vdots       & \vdots      & \vdots      & \ddots & \vdots  \\
			0            & \cdots      & 0      & \beta_{s+1}  & \beta_s     & \beta_{s-1}      &  \cdots      & \beta_1  \\
		\end{pmatrix}.
	\end{align*}
	
	To prove~B3, notice that the leftmost~$n-s$ columns of~$\boldB$ have leading coefficients in different positions, and hence they are linearly independent. Thus, the dimension of the column span of~$\boldB$ is at least~$n-s$, and since $\dim \cC=n-s$, the claim follows.

	To prove~B4, assume for contradiction that there exist a set of~$n-s$ linearly dependent rows. Hence, there exists a vector~$\boldv\in\bF^n$ of Hamming weight $n-s$ such that~$\boldv\boldB=0$. According to~B3, the columns of~$\boldB$ span~$\cC$, and hence the vector~$\boldv$ lies in the dual code~$\cC^\bot$ of~$\cC$. Since~$\cC^\bot$ is an~$[n,s]$ MDS code by Lemma~\ref{lemma:MDSproperties}, it follows that the minimum Hamming weight of a codeword in~$\cC^\bot$ is~$n-s+1$, a contradiction.
\end{proof}

	Since~$\cC^R$ is of dimension~$n-s$, it follows from parts~B2 and~B4 of Lemma~\ref{lemma:matrixB} that every set of~$n-s$ rows of~$\boldB$ are a basis to~$\cC^R$. Furthermore, since~$\1\in \cC$ it follows that~$\1\in \cC^R$. Therefore, there exists a function~$a:\cP(n)\to \bF^n$ such that for any set~$\cK\subseteq[n]$ of size~$n-s$ we have that~$\support{a(\cK)}\blue{\subseteq }\cK$ and $a(\cK)\cdot \boldB=\1$. 
	
	
	\begin{theorem}\label{theorem:cyclicMDS}.
		The above~$a$ and~$\boldB$ satisfy the EC condition (Definition~\ref{definition:ECcondition}).
	\end{theorem}

	In the remainder of this section, two cyclic MDS codes over the complex numbers and the real numbers are suggested, from which the construction in Theorem~\ref{theorem:cyclicMDS} can be obtained. These constructions are taken from~\cite{RealCodes} (Sec.~II.B), and are given with a few adjustments to our case. The contributions of these codes is summarized in the following theorem. 
	
	\begin{theorem}
	For any given~$n$ and~$s$ there exist explicit complex valued~$a$ and~$\boldB$ that satisfy the EC-condition with optimal~$d=s+1$. The respective encoding (i.e., constructing~$\boldB$) and decoding (i.e., constructing~$a(\cK)$ given~$\cK$) complexities are $O(s(n-s))$ and $O(s \log^2 s+n \log n)$, respectively. In addition, for any given~$n$ and~$s$ such that $n \ne s \bmod 2$ there exist explicit real valued~$a$ and~$\boldB$ that satisfy the EC-condition with optimal~$d=s+1$. The encoding and decoding complexities are~$O(\min\{s \log^2 s, n \log n\})$ and $O(g_s+s(n-s))$, where~$g_s$ is the complexity of inverting a generalized Vandermonde matrix.
	\end{theorem}
	
	\subsection{Cyclic-MDS Codes Over the Complex Numbers}\label{section:ComplexCyclic}
	For a given~$n$ and~$s$, let~$i=\sqrt{-1}$, and let $\cA\triangleq \{\alpha_j\}_{j=0}^{n-1}$ be the set of~$n$ complex roots of unity of order~$n$, i.e., $\alpha_j\triangleq e^{2\pi i j/n}$. Let~$\boldG\in\bC^{(n-s)\times n}$ be a complex Vandermonde matrix over~$\cA$, i.e., $g_{k,j}=\alpha_{j-1}^{k-1}$ for~$j\in[n]$ and any~$k\in[n-s]$. Finally, let~$\cC\triangleq\{\boldx\boldG\vert \boldx\in\bC^{n-s} \}$. It is readily verified that~$\cC$ is an~$[n,n-s]$ MDS code that contains~$\1$, whose codewords may be seen as the evaluations of all polynomials in~$\bC^{<n-s}[x]$ on the set~$\cA$.
	
	\begin{lemma} \label{lemma:MDScyclic}
		The code $\cC$ is cyclic.
	\end{lemma}
	\begin{proof}
	Let~$\boldc\in \cC$ be a codeword, and let $f_{\boldc}\in\bC^{<n-s}[x]$ be the corresponding polynomial. Consider the polynomial~$f_{\boldc'}(x)\triangleq f_{\boldc}(e^{2\pi i/n}\cdot x)$, and notice that $\deg f_{\boldc'}=\deg f_{\boldc}$. Further, it is readily verified that any~$j\in\{0,1,\ldots,n-1\}$ satisfies that $f_{\boldc'}(\alpha_j)=f_{\boldc}(\alpha_{j-1})$, where the indices are taken modulo~$n$. Hence, the evaluation of the polynomial~$f_{\boldc'}$ on the set of roots~$\cA$ results in the cyclic shift of the codeword~$\boldc$, and lies in~$\cC$ itself.
\end{proof}

\begin{corollary}
	The code~$\cC$ is a cyclic MDS code which contains~$\1$, and hence it can be used to obtain $a$ and~$\boldB$, as described in Theorem~\ref{theorem:cyclicMDS}.
\end{corollary}

Given a set~$\cK$ of~$n-s$ non-stragglers, an algorithm for computing the encoding vector~$a(\cK)$ in~$O(s\log^2s+n\log n)$ operations over~$\bC$ (after a one-time initial computation of~$O(s^2+s(n-s))$), is given in Appendix~\ref{section:AK_complex}.
The complexity of this algorithm is asymptotically smaller than the corresponding algorithm in~\cite{ShortDot} and~\cite{Wael} whenever~${s=o(n)}$. Furthermore, the cyclic structure of the matrix~$\boldB$ enables a very simple algorithm for its construction; this algorithm compares favorably with previous works for any~$s$, and is given in Appendix~\ref{section:AK_complex} as well.

\begin{remark}
Note that the use of complex rather than real matrix~$\boldB$ may potentially double the required bandwidth, since every complex number contains two real numbers. A simple manipulation of Algorithm~\ref{algorithm:distributedSGD} which resolves this issue is given in Appendix~\ref{section:BandwidthReduction}. This optimal bandwidth is also attained by the scheme in the next section, which uses a smaller number of multiplication operations. However, it is applicable only if~$n\ne s\bmod 2$.
\end{remark}

\subsection{Cyclic-MDS Codes Over the Real Numbers}\label{section:RealCyclic}
If one wishes to abstain from using complex numbers, e.g., in order to reduce bandwidth, we suggest the following construction, which provides a cyclic MDS code over the reals. This construction relies on~\cite{RealCodes} (Property~3), with an additional specialized property. 

	\begin{construction}\label{construction:realBCH}
		For a given~$n$ and~$s$ such that~$n\ne s\bmod 2$, define the following BCH codes over the reals. In both cases denote $\omega\triangleq e^{2\pi i/n}$.
		\begin{enumerate}
			\item If~$n$ is even and~$s$ is odd let~$s'\triangleq \floor{\frac{s}{2}}$, and let~$\cC_1$ be a BCH code which consists of all  polynomials in~$\bR^{<n}[x]$ that vanish over the set~$\cR_1\triangleq \{\omega^{n/2-s'},\omega^{n/2-s'+1},\ldots,\omega^{n/2+s'} \}$.
			\item If~$n$ is odd and~$s$ is even let~$n'\triangleq\floor{\frac{n}{2}}$, and let~$\cC_2$ be a BCH code which consists of all polynomials in~$\bR^{<n}[x]$ that vanish over the set~$\cR_2\triangleq \{\omega^{n'-s/2+1}, \omega^{n'-s/2+2},\ldots,\omega^{n'+s/2}  \}$.
		\end{enumerate}
	\end{construction}

	\begin{lemma}\label{lemma:BCHareCyclicMDS}
		The codes~$\cC_1$ and~$\cC_2$ from Construction~\ref{construction:realBCH} are cyclic~$[n,n-s]$ MDS codes that contain~$\1$.
	\end{lemma}

\begin{proof}
	According to Lemma~\ref{lemma:BCHcyclic}, it is clear that~$\cC_1$ and~$\cC_2$ are cyclic. According to the BCH bound (Theorem~\ref{theorem:BCHbound}), it is also clear that the minimum distance of~$\cC_1$ is at least~$|\cR_1|+1=s+1$, and the minimum distance of~$\cC_2$ is at least~$|\cR_2|+1=s+1$. Hence, to prove that~$\cC_1$ and~$\cC_2$ are MDS codes, it is shown that their code dimensions are~$n-s$.
	
	Since the sets~$\cR_1$ and~$\cR_2$ are closed under conjugation (i.e., $r$ is in~$\cR_i$ if and only if the conjugate of~$r$ is in~$\cR_i$) it follows that the polynomials $p_1(x)\triangleq \prod_{r\in \cR_1}(x-r)$ and $p_2(x)\triangleq \prod_{r\in \cR_2}(x-r)$ have real coefficients. Hence, by the definition of BCH codes it follows that 
	\begin{align}\label{equation:generatorPolys}
	\nonumber\{p_1(x)\cdot \ell(x)\vert \ell(x)\in\bR^{<n-s}[x] \}&\subseteq \cC_1\\
	\{p_2(x)\cdot \ell(x)\vert \ell(x)\in\bR^{<n-s}[x] \}&\subseteq C_2,
	\end{align}
	and hence, $\dim \cC_1 \ge n-s$ and $\dim \cC_2\ge n-s$. Let~$d(\cC_1)$ and~$d(\cC_2)$ be the minimum distances of~$\cC_1$ and~$\cC_2$, respectively, and notice that by the Singleton bound~\cite{Ronny} (Sec.~4.1) it follows that 
	\begin{align*}
	\dim \cC_1 &\le n - d(\cC_1)+1 \le n- (s+1)+1=n-s\\
	\dim \cC_2 &\le n - d(\cC_2)+1 \le n- (s+2)+1=n-s,
	\end{align*}
	and thus~$\cC_1$ and~$\cC_2$ satisfy the Singleton bound with equality, or equivalently, they are MDS codes. To prove that~$\1$ is in~$\cC_1$ and~$\cC_2$, we ought to show that~$\1(r_1)=\1(r_2)=0$ (where~$\1$ denotes the all ones polynomial here) for every~$r_1\in R_1$ and~$r_2\in R_2$, which amounts to showing that~$\sum_{j=0}^{n-1}r_1^j=\sum_{j=0}^{n-1}r_2^j=0$. It is well-known that the sum of the~$0$'th to~$(n-1)$'th power of any root of unity of order~$n$, other than~$1$, equals zero. Since
	\begin{align*}
	1 &\leq \frac{n}{2} -\bigfloor{\frac{s}{2}}=\frac{n}{2}-s'<\frac{n}{2}-s'+1\\
	&<\cdots<\frac{n}{2}+s'=\frac{n}{2}+\bigfloor{\frac{s}{2} } \leq n-1 \mbox{ if~$n$ is even and~$s$ is odd, and}\\
	1&\le \bigfloor{\frac{n}{2} }-\frac{s}{2}+1=n'-\frac{s}{2}+1<n'-\frac{s}{2}+2\\
	&<\cdots<n'+\frac{s}{2}=\bigfloor{\frac{n}{2} }+\frac{s}{2}\le n-1\mbox{ otherwise,}
	\end{align*}	
	it follows that \blue{all powers of~$\omega$ in~$\cR_1$ and~$\cR_2$ are between~1 and~$n-1$, and hence}~$1\notin \cR_1$ and that~$1\notin \cR_2$. Hence, we have that~$\1(r_1)=\1(r_2)=0$ for every~$r_1\in \cR_1$ and~$r_2\in \cR_2$, which concludes the claim.
\end{proof}

Algorithms for computing the matrix~$\boldB$ and the vector~$a(\cK)$ for the codes in this subsection are given in Appendix~\ref{section:AK_real}.
The algorithm for construction~$\boldB$ outperforms previous works whenever~$s=o(n)$, and the algorithm for computing~$a(\cK)$ outperforms previous works for a smaller yet wide range of~$s$ values.

\section{Approximate Gradient Coding from Expander Graphs}\label{section:approximate}
Recall that in order to retrieve that exact gradient, one must have~$d\ge s+1$, an undesirable overhead in many cases. To break this barrier, we relax the requirement to retrieve the gradient exactly, and settle for an approximation of it. Note that trading the exact gradient for an approximate one is a necessity in many variants of gradient descent (such as the acclaimed \textit{stochastic gradient descent}~\cite[Sec.~14.3]{MachineLearningI}), and hence our techniques are aligned with common practices in machine learning.

\purple{For a set~$\cK\subseteq [n]$ of non-stragglers, let~$\1_\cK$ be its binary characteristic vector with respect to~$[n]$. Setting~$\boldB$ as the identity matrix and~$a$ as the function which maps~$\cK\in\cP(n)$ to~$\frac{n}{|\cK|}\cdot\1_\cK$ corresponds to executing gradient descent while \textit{ignoring} the stragglers~\cite{chen2016revisiting}, and averaging the partial gradients from~$\cK$ (the factor~$\frac{n}{|\cK|}$ corrects the factor~$\frac{1}{n}$ in~\eqref{equation:N(w)}). In what follows, this is referred to as the \textit{trivial scheme}. The respective~$\boldB$ and~$a$ clearly satisfy the~$\epsilon$-AC scheme for~$\epsilon(s)=\sqrt{\frac{ns}{n-s}}$, where~$s=s(\cK)\triangleq n-|\cK|$, since
	\begin{align}\label{equation:trivialApproximationScheme}
	d_2(a(\cK)\boldB,\1)=d_2(\tfrac{n}{n-s}\1_\cK,\1)=\sqrt{\sum_{j\in\cK}(\tfrac{n}{n-s}-1)^2+\sum_{j\notin\cK}1^2}=\sqrt{\frac{ns}{n-s}}.
	\end{align}} 
We show that this can be outperformed by setting~$\boldB$ to be a normalized adjacency matrix of a connected regular graph on~$n$ nodes, which is constructed by the master before dispersing the data, and setting~$a$ to be some simple function. 
	
	The resulting error function~$\epsilon(s)$ depends on the parameters of the graph, whereas the resulting storage overhead~$d$ is given by its \textit{degree} (i.e., the \textit{fixed} number of neighbors of each node). The error function is given below for a general connected and regular graph, and particular examples with their resulting errors are given in the sequel. In particular, it is shown that taking the graph to be an expander graph provides an error term~$\epsilon$ which is \blue{smaller than~\eqref{equation:trivialApproximationScheme} for any~$s$}.
	
	For a given~$n$ let~$G$ be a connected~$d$-regular graph on~$n$ nodes, with eigenvalues~$\lambda_1\ge\ldots\ge \lambda_n$, corresponding (\blue{row}) eigenvectors~$\boldv_1=\1,\boldv_2,\ldots,\boldv_n$, and $\lambda\triangleq \max\{|\lambda_2|,|\lambda_n| \}$ as described in Subsection~\ref{section:notions}. For a given~$\cK\subseteq[n]$  
	define~$\boldu_\cK\in \bR^n$ as
	\begin{align}\label{equation:uK}
		(\boldu_\cK)_i=
		\begin{cases}
			-1 & i\notin \cK\\
			\frac{s}{n-s} & i\in \cK
		\end{cases},
	\end{align}
	\blue{and let~$\Span{\boldv_1,\boldv_2,\ldots}$ be the span of~$\boldv_1,\boldv_2,\ldots$ over~$\bR^n$.}
	
	\begin{lemma}\label{lemma:uK}
		For any~$\cK\subseteq[n]$ we have~$\boldu_\cK\in\Span{\boldv_2,\ldots,\boldv_n}$.
	\end{lemma}

\begin{proof}
	First, observe that $\Span{\boldv_2,\ldots,\boldv_n}$ is exactly the subspace of all vectors whose sum of entries is zero. This follows from the fact that $\{\1,\boldv_2,\ldots,\boldv_n\}$ is an orthogonal basis, hence $\boldv_i^\top\cdot\1=0$ for every~$i\ge 2$, and from the fact that~$\{\boldv_2,\ldots,\boldv_n\}$ are linearly independent. Since the sum of entries of~$\boldu_\cK$ is zero, the result follows.
\end{proof}

\begin{corollary}\label{corollary:uK}
	For any~$\cK\subseteq[n]$ there exists $\alpha_2,\ldots,\alpha_n\in\bR$ such that $\boldu_\cK=\alpha_2\boldv_2+\ldots+\alpha_n\boldv_n$, and $\norm{\boldu_\cK}=\sqrt{\sum_{i=2}^{n}\alpha_i^2}=\sqrt{\frac{ns}{n-s}}$.
\end{corollary}

\begin{proof}
	The first part follows immediately from Lemma~\ref{lemma:uK}. The second part follows by computing the~$\ell_2$ norm of~$\boldu_\cK$ in two ways, once by its definition~\eqref{equation:uK} and again by using the representation of~$\boldu_\cK$ as a linear combination of the orthonormal set~$\{\boldv_2,\ldots,\boldv_n \}$.
\end{proof}

Now, let~$\boldB\triangleq \frac{1}{d}\cdot \boldA_G$, define~$a:\cP(n)\to \bR^n$ as~$a(\cK)=\boldu_\cK+\1$, and observe that~$\support{a(\cK)}=\cK$ for all~$\cK\in\cP(n)$. Note that computing~$a(\cK)$ given~$\cK$ is done by a straightforward~$O(n)$ algorithm. The error function~$\epsilon$ is given by the following lemma.

\begin{lemma}\label{lemma:ABproximity}
	For every \blue{nonempty\footnote{For~$\cK=\varnothing$ we clearly have~$\boldu_\cK+\1=0$, and hence~$d_2(a(\cK)\boldB,\1)=\norm{\1}=\sqrt{n}$.}} set~$\cK\subseteq[n]$ of size~$n-s$ we have that $d_2(a(\cK)\boldB,\1)\le \frac{\lambda}{d}\cdot \sqrt{\frac{ns}{n-s}}\triangleq\epsilon(s)$.
\end{lemma}

\begin{proof}
	Notice that the eigenvalues of~$\boldB$ are $\mu_i\triangleq \frac{\lambda_i}{d}$, and hence~$\mu\triangleq \max\{|\mu_2|,|\mu_n| \}$ equals~$\frac{\lambda}{d}$. Further, the eigenvectors are identical to those of~$\boldA_G$. Therefore, it follows from Corollary~\ref{corollary:uK} that
	\begin{align*}
	d_2(a(\cK)\boldB,\1)&=d_2((\1+\boldu_\cK) \boldB,\1)\\
	&=d_2((\1+\alpha_2 \boldv_2+\ldots+\alpha_n \boldv_n) \boldB,\1)\\
	&=d_2(\1+\alpha_2 \mu_2\boldv_2+\ldots+\alpha_n\mu_n \boldv_n,\1)\\
	&=\norm {\alpha_2 \mu_2\boldv_2+\ldots+\alpha_n\mu_n \boldv_n},
	\end{align*}
	and since $\{\boldv_2,\ldots,\boldv_n \}$ are orthonormal, it follows that
	\begin{align*}
		\norm {\alpha_2 \mu_2\boldv_2&+\ldots+\alpha_n\mu_n \boldv_n}\\
		&=\sqrt{\sum_{i=2}^{n}\mu_i^2\alpha_i^2}\le \sqrt{\sum_{i=2}^{n}\mu^2\alpha_i^2}\\
		&=\mu\sqrt{\sum_{i=2}^{n}\alpha_i^2}=\frac{\lambda}{d}\sqrt{\frac{ns}{n-s}}.\qedhere
	\end{align*}
\end{proof}

\begin{corollary}
	The above~$a$ and~$\boldB$ satisfy the~$\epsilon$-AC condition for~$\epsilon(s)=\frac{\lambda}{d}\sqrt{\frac{ns}{n-s}}$. The storage overhead of this scheme equals the degree~$d$ of the underlying regular graph~$G$.
\end{corollary}


It is evident that in order to obtain small deviation~$\epsilon(s)$, it is essential to have a small~$\lambda$ and a large~$d$. However, most constructions of expanders have focused in the case were~$d$ is constant (i.e., $d=O(1)$). On one hand, constant~$d$ serves our purpose well \blue{in terms of storage overhead}, since it implies a \textit{constant} storage overhead. On the other hand, a constant~$d$ does not allow~$\lambda/d$ to tend to zero as~$n$ tends to infinity due to Theorem~\ref{theorem:lambdaLowerBound}.

\purple{It is readily verified that our scheme outperforms the trivial one by a multiplicative factor of~$\frac{\lambda}{d}$, which is less than one for every non-bipartite graph; bipartite graphs are used in a slightly different fashion in the sequel. We conclude the discussion with several examples,}
%
%
%
the first of which uses \textit{Margulis graphs} (\cite{Expanders1}, Sec.~8), that are rather easy to construct. 
\begin{example}
	For any integer~$n$ there exists an $8$-regular graph on~$n$ nodes with $\lambda \le 5\sqrt{2}$. For example, by using these graphs with the parameters~$n=500$, $d=8$, we have \purple{an improvement factor of~$\frac{\lambda}{d}=\frac{5\sqrt{2}}{8}\approx 0.883$.}
\end{example}

Several additional examples for Ramanujan graphs, which attain \purple{an improvement factor~$\le\frac{2\sqrt{d-1}}{d}\approx\frac{2}{\sqrt{d}}$} but are harder to construct, are as follows.

\begin{example}\cite{RamanujanGraphs}
	Let~$p$ and~$q$ be distinct primes such that~$p=1\bmod 4$, $q=1\bmod 4$, and such that the Legendre symbol $\left(\frac{p}{q}\right)$ is~$1$ (i.e.,~$p$ is a quadratic residue modulo~$q$). Then, there exist a non-bipartite Ramanujan graph on~$n=\frac{q(q^2-1)}{2}$ nodes and constant degree~$p+1$. 
	\begin{enumerate}
		\item If~$p=13$ and~$q=17$ then~$n=2448$, $d=14$, and~$\frac{2\sqrt{d-1}}{d}\approx 0.534$. 
		
		\item If~$p=5$ and~$q=29$ then~$n=12180$, $d=6$, and~$\frac{2\sqrt{d-1}}{d}\approx 0.816$.
	\end{enumerate}
	
\end{example}

Restricting~$d$ to be a constant (i.e., not to grow with~$n$) is detrimental to the improvement factor~$\frac{\lambda}{d}$ 
due to Theorem~\ref{theorem:lambdaLowerBound}, but allows lower storage overhead. If one wishes a smaller \purple{(i.e., better) improvement factor} at the price of higher overhead, the following is useful.

\begin{example}\cite{BiluLinial}
	There exists a polynomial algorithm (in~$n$) to produce a graph~$G$ with the parameters~$(n,d,\lambda)=(2^m,m-1,\sqrt{m\log^3 m})$. For this family of graphs, the term~$\frac{\lambda}{d}$ goes to zero as~$n$ goes to infinity.
\end{example}


\subsection{Bipartite expanders.}\label{section:bipartite}

The above approximation scheme can be used with a \textit{bipartite} graph~$G$ as well. However, bipartite graphs satisfy that~$\lambda=d$, and hence the resulting error function~$\epsilon(s)=\sqrt\frac{ns}{n-s}$ is identical to the error function of the trivial scheme~\eqref{equation:trivialApproximationScheme}, which requires lower overhead, and hence no gain is attained. However, in what follows it is shown that bipartite graphs on~$2n$ nodes can be employed in a slightly different fashion, and obtain~\blue{$\epsilon(s)=\frac{\sigma_2}{d}\sqrt\frac{ns}{n-s}$ for some~$\sigma_2<d$ that is defined shortly}. 


To this end, we require the notion of \textit{singular value decomposition}, which implies that any matrix~$\boldP\in\bR^{n\times n}$ can be written as~$\boldP=\boldU\boldD\boldV^\top$, where~$\boldD\in\bR^{n\times n}$ is a diagonal matrix, and~$\boldU$ and~$\boldV$ are orthonormal matrices (i.e., $\boldU\boldU^\top=\boldV\boldV^\top=\boldI_n$, where~$\boldI_n$ is the identity matrix of order~$n$). The elements~$\{\sigma_i\}_{i=1}^n$ on the diagonal of~$\boldD$, which are nonnegative, are called the \textit{singular values} of~$\boldP$, the columns~$\{\boldu^\top_i\}_{i=1}^n$ of~$\boldU$ are called \textit{left-singular vectors} of~$\boldP$, and the columns~$\{\boldv^\top_i\}_{i=1}^n$ of~$\boldV$ are called \textit{right-singular vectors} of~$\boldP$. The singular values and singular vectors of~$\boldP$ can be arranged in triples~$\{(\boldu_i,\boldv_i,\sigma_i)\}_{i=1}^n$ such that for all~$i\in[n]$ we have~$\boldP \boldv_i^\top =\sigma_i \boldu_i^\top$ and $\boldP^\top \boldu^\top_i=\sigma_i \boldv^\top_i$, which implies that~$\boldP^\top \boldP \boldv^\top_i=\sigma_i^2 \boldv^\top_i$ and~$\boldP\boldP^\top \boldu^\top_i=\sigma_i^2\boldu^\top_i$.

Let~$G=(\cL\cup \cR,\cE)$ be a $d$-regular (in both sides) and connected bipartite graph on~$2n$ nodes (and hence $|\cL|=|\cR|=n$), with an  adjacency matrix 
\begin{align*}
\boldA_G = 
\begin{pmatrix}
0 & \boldC \\
\boldC^\top & 0
\end{pmatrix}
\end{align*}
for some~$n\times n$ real matrix~$\boldC$, and eigenvalues~$\lambda_1=d\ge \lambda_2 \ge \cdots \ge \lambda_{2n}=-d$. Let~$\{(\boldu_i,\boldv_i,\sigma_i)\}_{i=1}^n$ be the set of triples of left-singular vectors, right-singular vectors, and singular values of~$\boldC$, as explained above, where~$\sigma_1\ge \cdots\ge \sigma_n\ge 0$. The next \blue{well-known} lemma presents the connection between the singular values of~$\boldC$ and the eigenvalues of~$\boldA_G$. \blue{Its proof is a combination of a few simple exercises, and is given for completeness.}

\begin{lemma}\label{lemma:LambdaSigma}
	$\{\lambda_i\}_{i=1}^{2n}=\{ \sigma_i \}_{i=1}^n\cup \{ -\sigma_i \}_{i=1}^n$.
\end{lemma}
\begin{proof}
	For any~$i\in[n]$ we have that
	\begin{align*}
	(\boldu_i,\boldv_i)\begin{pmatrix}
	0 & \boldC \\ \boldC^\top & 0
	\end{pmatrix}=(\boldv_i \boldC^\top, \boldu_i \boldC)=\sigma_i(\boldu_i,\boldv_i),
	\end{align*}
	and hence~$\sigma_i\in\{ \lambda_i \}_{i=1}^{2n}$. It is an easy exercise to verify that the eigenvalues of a bipartite graph are symmetric around zero, and hence it follows that~$-\sigma_i\in\{ \lambda_i \}_{i=1}^{2n}$ as well.
	
	Conversely, for any~$i\in[2n]$ let~$(\bolds_i,\boldt_i)$ be an eigenvector of~$\boldA_G$ with a corresponding eigenvalue~$\lambda_i$, where~$\bolds_i,\boldt_i\in\bR^n$. Since~$(\bolds_i,\boldt_i)\boldA_G=\lambda_i(\bolds_i,\boldt_i)$, it follows that
	\begin{align*}
	\left\{ 
	\begin{array}{ll}
	\boldt_i \boldC^\top &= \lambda_i \bolds_i\\
	\bolds_i \boldC      &= \lambda_i \boldt_i
	\end{array} \right.&\mbox{, and hence }
	\left\{ 
	\begin{array}{ll}
	\boldt_i \boldC^\top \boldC &= \lambda_i^2 \boldt_i\\
	\bolds_i \boldC \boldC^\top &= \lambda_i^2 \bolds_i
	\end{array} \right..
	\end{align*}
	Therefore, it follows that~$\lambda_i^2\in\{ \sigma_i^2 \}_{i=1}^n$, and thus~$\lambda_i \in\{ \sigma_i \}_{i=1}^n\cup \{ -\sigma_i \}_{i=1}^n$.
\end{proof}

From Lemma~\ref{lemma:LambdaSigma}, the following corollaries are easy to prove.

\begin{corollary}\label{corollary:singularVectors}
	~
	\begin{itemize}
		\item [A.] $(\1,\1,d)\in\{ (\boldu_i,\boldv_i,\sigma_i) \}_{i=1}^n$.
		\item [B.] \blue{$\sigma_2< d$}.
	\end{itemize}
\end{corollary}
\begin{proof}
	~
	\begin{itemize}
		\item [A.] Repeat the second part of the proof of Lemma~\ref{lemma:LambdaSigma} with~$i=1$, namely, with~$(\bolds_i,\boldt_i)=(\1,\1)$ and~$\lambda_i=d$.
		\item [B.] Since~$\sigma_1\ge \ldots\ge \sigma_n$, it follows from Lemma~\ref{lemma:LambdaSigma} that \blue{$\lambda_2= \sigma_2$}. Further, since~$G$ is connected, it follows that~$\lambda_2<d$, and hence~$\sigma_2<d$.\qedhere
	\end{itemize}
\end{proof}

Given Corollary~\ref{corollary:singularVectors}.A, we may assume without loss of generality that~$\boldu_1=\boldv_1=\1$, and that~$\norm{\boldv_i}=\norm{\boldu_i}=1$ for every~$i\in\{2,\ldots,n\}$. Now, for a given set~$\cK\subseteq[n]$ of size~$n-s$, let~$\boldu_\cK$ and~$a(\cK)$ be as in Section~\ref{section:approximate}. 
Notice that since~$\{ \boldu_i \}_{i=1}^n$ is an orthonormal basis, and since~$\1$ is orthogonal to~$\boldu_\cK$, it follows that~$\boldu_\cK\in \Span{\boldu_2,\ldots,\boldu_n}$. 
Hence, there exist real~$\alpha_2,\ldots,\alpha_n$ such that~$\boldu_\cK=\sum_{i=2}^{n}\alpha_i\boldu_i$. Further, it follows that
\begin{align*}
\norm{\boldu_\cK}=\sqrt{\frac{ns}{n-s}}=\sqrt{\sum_{i=2}^{n}\alpha_i^2}.
\end{align*}
By setting~$\boldB\triangleq\frac{1}{d}\boldC$, we have the following lemma, which may be seen as the equivalent of Lemma~\ref{lemma:ABproximity} to the bipartite case.

\begin{proposition}
	For every set~$\cK\subseteq[n]$ of size~\blue{$n-s$}, $d_2(a(\cK)\boldB,\1)\le \frac{\blue{\sigma_2}}{d}\cdot \sqrt{\frac{ns}{n-s}}$.
\end{proposition}

\begin{proof}
	By Corollary~\ref{corollary:singularVectors}, and since~$\{\boldv_i \}_{i=2}^n$ is an orthonormal set, it follows that
	\begin{align*}
	d_2(a(\cK)\boldB,\1)&=\norm{(\1+\boldu_\cK)\cdot\tfrac{1}{d}\boldC-\1}=\norm{\tfrac{1}{d} \boldu_\cK \boldC}\\
	&=\norm{\tfrac{1}{d} \left( \sum_{i=2}^{n}\alpha_i \boldu_i \right)\boldC}=\norm{\tfrac{1}{d}  \sum_{i=2}^{n}\alpha_i\sigma_i \boldv_i  }\\
	&= \frac{1}{d}\sqrt{ \sum_{i=\blue{2}}^{n}\sigma_i^2\alpha_i^2 }\le\frac{\sigma_2}{d}\sqrt{\frac{ns}{n-s}}\qedhere
	\end{align*}
\end{proof}

Applying the above lemma on several constructions of bipartite expanders provides the following examples.
\begin{example}
	Let~$p$ and~$q$ be distinct primes such that~$p=1\bmod 4$, $q=1\bmod 4$, and such that the Legendre symbol $\left(\frac{p}{q}\right)$ is~$-1$. Then, there exists a bipartite Ramanujan graph on~$2n=q(q^2-1)$ nodes with~$\blue{\sigma_2}\le 2\sqrt{p}$.
	\begin{enumerate}
		\item If~$p=5$ and~$q=13$ then~$n=\frac{2184}{2}=1092$, $d=6$, and~$\frac{\sigma_2}{d}\le \frac{2\sqrt{5}}{6}\approx 0.745$.
		\item If~$p=13$ and~$q=5$ then~$n=\frac{120}{2}=60$, $d=14$, and~$\frac{\sigma_2}{d}\le \frac{2\sqrt{13}}{14}\approx 0.515$.
	\end{enumerate}
\end{example}

\subsection{Lower bound.}

Finally, we have the following lower bound on the approximation error of any \textit{Approximate Computation (AC)} scheme, that establishes asymptotic optimality of our scheme, up to constants, when used with Ramanujan graphs. In what follows, for any set of vertices~$\cA$ in a graph~$G$, let $\mathcal{N}(\cA)$ be the set of vertices which has at least one neighbor in~$\cA$, and for a vertex~$V$ let $d(V)$ denote its degree.

\blue{
\begin{lemma}\label{lemma:missingLemma}
	Let~$G=(\cW\cup\cP,\cE)$ be a bipartite graph where~$|\cW|=|\cP|=n$ for some~$n$ and $d(W)\le d$ for every $W\in\cW$ and some~$d$. Then, for every~$r\in\{1,2,\ldots, \floor{\frac{n}{d}}\}$ there exists a set~$\cQ_r\subseteq \cP$ of size~$r$ such that~$|\cN(\cQ_r)|\le d|\cQ_r|$.
\end{lemma}
\begin{proof}
	Let~$\cP=\{P_1,\ldots,P_n\}$ and without loss of generality assume that~$d(P_1)\le d(P_2)\le\ldots\le d(P_n)$. Due to this ordering of~$P_1,\ldots,P_n$, and due to the bounded degree of vertices in~$\cW$, it follows that $d\ge\text{avg}_n\ge\text{avg}_{n-1} \ge \text{avg}_1$, where~$\text{avg}_j\triangleq \frac{1}{j}\sum_{i=1}^j d(P_i)$ for every~$j\in [n]$. Therefore, the average degree of~$\cQ_r\triangleq\{
	P_1,\ldots,P_r\}$ is at most~$d$, which implies that~$\sum_{j=1}^rd(P_i)\le rd$. Since the size of~$\cN(\cQ_r)$ is at most the sum of degrees in~$\cQ_r$, it follows that~$|\cN(\cQ_r)|\le \sum_{j=1}^rd(P_i)\le rd=d|\cQ_r|$.
\end{proof}
}

\begin{lemma}\label{lemma:lowerBound}
	Consider any~$\boldB\in \bR^{n\times n}$ with each row having at most~$d$ non-zeros. Then, for any~$s$ such that~$d<s<n$ there exists a set~$\cK\subseteq[n]$ of size~$n-s$ such that
    \begin{equation}
	\min_{\substack{\bolda\in \bR^n \\ \support{a}\subseteq \cK}} d_2(\bolda \boldB,\1) \geq \sqrt{\left\lfloor \frac{s}{d} \right\rfloor}.
	\end{equation}
\end{lemma}

\begin{proof}
	Associate a bipartite graph $G=(\mathcal{W}\cup\mathcal{P},\cE)$ with~$\boldB$ as follows. Consider left vertices $\mathcal{W} = \{W_1, W_2, \ldots, W_n\}$ corresponding to workers, and right vertices $\mathcal{P} = \{P_1, \ldots, P_n\}$ corresponding to data parts. We draw an edge $\{W_i,P_j\}\in \cE$ between worker~$W_i$ and data part $P_j$ if $(\boldB)_{i,j}\ne 0$. According to Lemma~\ref{lemma:missingLemma}, for any~$s$ such that~$d<s<n$ there exists a set~$\cQ\subseteq \cP$ of size~$\floor{\frac{s}{d}}$ such that~$|\cN(\cQ)|\le d\cdot \floor{\frac{s}{d}}\le s$. 
	
	\blue{Now, let~$\cK$ be any set of size $n-s$ which is contained in~$\cW\setminus \cN(\cQ)$; this set exists since~$|\cW\setminus\cN(\cQ)|=n-|\cN(\cQ)|\ge n-s$. In addition, let $\bolda=(a_j)_{j=1}^n\in\bR^n$ be any vector with~$\support{\bolda}\subseteq \cK$. Notice that~$(\boldB)_{\ell,j}=0$ whenever~$\ell\notin\cN(\cQ)$ and~$j\in \cQ$, and~$a_\ell=0$ whenever~$\ell\in\cN(\cQ)$. Thus, for~$j\in\cQ$ we have
		\begin{align*}
			(\bolda\boldB)_j=\sum_{\ell=1}^n a_\ell(\boldB)_{\ell,j}=\sum_{\ell\in\cN(\cQ)}a_\ell(\boldB)_{\ell,j}+\sum_{\ell\notin\cN(\cQ)}a_\ell\boldB_{\ell,j}=0,
		\end{align*}
		since in the summation over~$\ell\in\cN(\cQ)$ we have~$a_\ell=0$, and in the summation over~$\ell\notin\cN(\cQ)$ we have~$(\boldB_{\ell,j})=0$. Therefore, $(\bolda\boldB-\1)_j=-1$ for~$j\in \cQ$, and thus~$d_2(\bolda\boldB,\1)\ge \sqrt{|\cQ|} \geq \sqrt{\left\lfloor \frac{s}{d}\right\rfloor}$, which concludes the claim.}
\end{proof}

The above lemma establishes the asymptotic optimality (up to constants) of our scheme, when used with Ramanujan graphs. Recall that for Ramanujan graphs we have $\lambda \leq 2\sqrt{d}$. Thus, using our proposed scheme, we get an~$a$ and~$\boldB$ that satisfy the~$\epsilon$-AC condition for~$\epsilon(s)\leq \frac{2}{\sqrt{d}}\sqrt{\frac{s}{1-(s/n)}}$ which tends to  $2\sqrt{\frac{s}{d}}$ as $s/n \to 0$.
\color{black}
\subsection{A few remarks about convergence}\label{section:aFewRemarks}
While our results in Section~\ref{section:exact} guarantee the exact computation of the gradient, and hence the convergence of the overall gradient descent algorithm, some care is needed to guarantee convergence when the gradient is approximated, as done in Section~\ref{section:approximate} and Section~\ref{section:bipartite}. We operate under the standard assumption (e.g.~~\cite{ErasureHead}) that the arrival of a computation result from a server, within a predefined time frame, is modeled by a Bernoulli random variable.
Then, the convergence of our algorithms can be guaranteed as a special case of the SGD algorithm~\cite[Sec.~14.3]{shalev2014understanding}; this applies for Section~\ref{section:approximate} and Section~\ref{section:bipartite}, but for brevity we focus on Section~\ref{section:approximate}. Furthermore, the $\ell_2$-deviations that were discussed above provide better convergence guarantees than the trivial algorithm, that can also be seen as a special case of SGD under identical assumptions.

Assume that every server in the system responds by some given deadline with probability~$q$, i.e., it is a straggler with probability~$1-q$. Under this premise, it is shown that the expected value of~$\boldv_r$ in iteration~$r$ of Algorithm~\ref{algorithm:distributedSGD}, when employing the expander scheme from Section~\ref{section:approximate}, is the true gradient~$\nabla L_\cS(\boldw^{(r)})$ up to a constant factor. Since the analysis is identical for every~$r$, we omit this superscript from now on.

Let~$X_i$ be a $\text{Bernoulli}(q)$ random variable that equals~$1$ if the~$i$'th server responded, and~$0$ otherwise. Further, if no server responded then~$\boldv=0$, and hence we have that $\boldv=\boldz\cdot \boldB\cdot\boldN(\boldw)$, where~$\boldz$ is a random variable such that
\begin{align*}
\boldz=
\begin{cases}
0 & \mbox{if }\sum_{i=1}^nX_i=0\\
\frac{n}{\sum_{i=1}^nX_i}(X_1,\ldots,X_n)	&\mbox{else}.
\end{cases}
\end{align*}
Since~$\boldB$ and~$\boldN(\boldw)$ are constants, it suffices to show that~$\bE\boldz=c\1$ for some constant~$c$. Then, we would have that $$\bE \boldv=(\bE \boldz)\cdot \boldB\boldN(\boldw)=c\1\boldB\boldN(\boldw)=c\1\boldN(\boldw)=c \nabla L_\cS(\boldw),$$
and thus convergence of the algorithm can be guaranteed as a special case of the SGD algorithm, by adjusting the learning rate properly.
\begin{lemma}
	For~$\boldz$ and~$q$ defined as above, we have that~$\bE\boldz=(1-(1-q)^n)\cdot \1$.
\end{lemma}
\begin{proof}
We have that
\begin{align*}
\bE\boldz=0\cdot (1-q)^n+\sum_{\cF\subseteq [n],\cF\ne\varnothing}q^{|\cF|}(1-q)^{n-|\cF|}\cdot\tfrac{n}{|\cF|}\cdot\1_{\cF},
\end{align*}
where~$\1_\cF$ is the indicator vector of a set~$\cF$. Therefore, for every~$j\in[n]$, the $j$'th entry of~$\bE\boldz$ equals
\begin{align*}
\sum_{\cF\subseteq[n],\cF\ne\varnothing}q^{|\cF|}(1-q)^{n-|\cF|}\cdot \tfrac{n}{|\cF|}\cdot \1(j\in\cF),
\end{align*}
where~$\1(j\in\cF)$ is a $0$-$1$ indicator for the truth value of ``$j\in\cF$''. Re-indexing the above sum by all possible sizes of a set~$\cF$ that contains~$j$, we have
\begin{align*}
\sum_{i=1}^n \binom{n-1}{i-1}q^i(1-q)^{n-i}\cdot \tfrac{n}{i}&=\sum_{i=1}^n \binom{n}{i}q^i(1-q)^{n-i} \\
&=\sum_{i=0}^n \binom{n}{i}q^i(1-q)^{n-i}-(1-q)^n=1-(1-q)^n,
\end{align*}
where the last inequality follows from the Binomial Theorem. Hence, it follows that~$\bE\boldz=(1-(1-q)^n)\1$, which concludes the claim.
\end{proof}

It is widely known that if the gradient update in gradient descent is taken from a probability distribution whose expectation is the true gradient, then convergence is guaranteed~\cite[Sec.~14.3]{shalev2014understanding}. Therefore, the algorithm in Section~\ref{section:approximate} converges. According to the same argument, it can be readily verified that the trivial algorithm~\eqref{equation:trivialApproximationScheme} converges as well. However, in what follows it is shown that the improved~$\ell_2$-deviation guarantees of our algorithm, with respect to the trivial one, provide faster convergence for an important family of functions.

For~$\beta\in\bR$, a continuously differentiable function~$f$ is $\beta$-smooth in a domain~$\cX$ if~$\norm{\nabla f(\boldx)-\nabla f(\boldy)}\le\beta\norm{\boldx-\boldy}$ for every~$\boldx,\boldy\in\cX$. The following theorem applies whenever~$L_\cS$ is $\beta$-smooth, and the learning algorithm is constrained to some domain~$\cX$. For the sake of citing the following theorem, notice that the vector~$\boldv$ in Algorithm~\ref{algorithm:distributedSGD} is a (stochastic) function of~$\boldw$, and thus we write~$\boldv(\boldw)$. In addition, we let~$\boldw^*$ be the minimizer of~$L_\cS$ in~$\cX$.
\begin{theorem}(\cite[Theorem~6.3]{bubeck2015convex}\footnote{This theorem is stated in much greater generality in~\cite{bubeck2015convex}, but cited herein as a special case for simplicity.})\label{theorem:Bubeck}
	For~$R^2\triangleq \sup_{\boldw\in\cX}\tfrac{1}{2}\norm{\boldw}^2-\norm{\boldw^{(1)}}$, if~$\bE\norm{\nabla L_\cS(\boldw)-\boldv(\boldw)}^2\le \sigma^2$ for every~$\boldw\in\cX$ and~$\bE\boldv(\boldw)=c\nabla L_\cS(\boldw)$ for some~$c\in\bR$, then Algorithm~\ref{algorithm:distributedSGD} with step size~$\frac{1}{c(\beta+1/\eta)}$ and~$\eta\triangleq \frac{R}{\sigma}\sqrt{\frac{2}{t}}$ satisfies
	\begin{align*}
		\bE L_\cS\left( \frac{1}{t} \sum_{s=1}^t \boldw^{(s+1)} \right)-L_\cS(\boldw^*)\le R\sigma\sqrt{\frac{2}{t}}+\frac{\beta R^2}{t}.
	\end{align*}
\end{theorem}
By performing some technical calculations, that are given in full in Appendix~\ref{appendix:convergenceProof}, we have that
\begin{align*}
	\bE\norm{\nabla L_\cS(\boldw)-\boldv(\boldw)}^2&\le \specnorm{\boldN(\boldw)}^2\cdot n\cdot\left( (1-q)^n+\tfrac{2(1-q)}{q} \right) \mbox{ for the trivial scheme; and}\\
	\bE\norm{\nabla L_\cS(\boldw)-\boldv(\boldw)}^2&\le \specnorm{\boldN(\boldw)}^2\cdot n\cdot\left( (1-q)^n+\tfrac{\lambda^2}{d^2}\cdot\tfrac{2(1-q)}{q} \right)\mbox{ for our scheme}.
\end{align*}
Therefore, Theorem~\ref{theorem:Bubeck} and the analysis in Section~\ref{section:approximate} show that our algorithm provides better bounds than the trivial algorithm, and thus it is likely to converge \textit{faster} in many cases.

\color{black}
\section{Experimental Results}\label{section:experimentalResults}
\begin{figure*}[t]
	\centering
	\begin{subfigure}[b]{0.3\textwidth}
		\includegraphics[width=\textwidth]{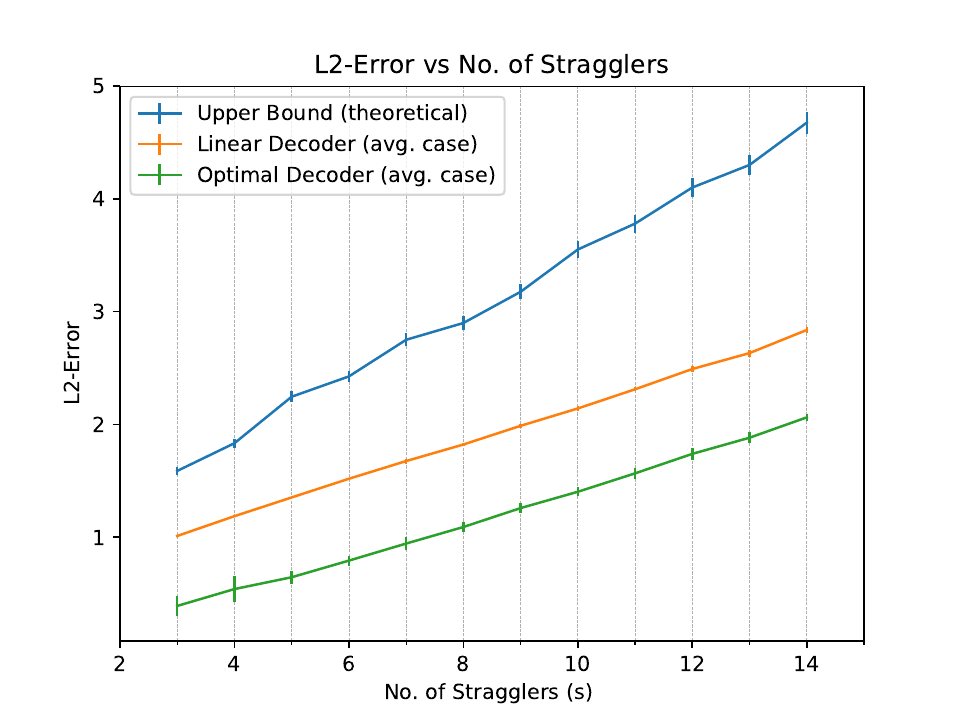}
		\caption{$n=30$, $d=3$}
		\label{fig:l2a}
	\end{subfigure}
	\begin{subfigure}[b]{0.3\textwidth}
		\includegraphics[width=\textwidth]{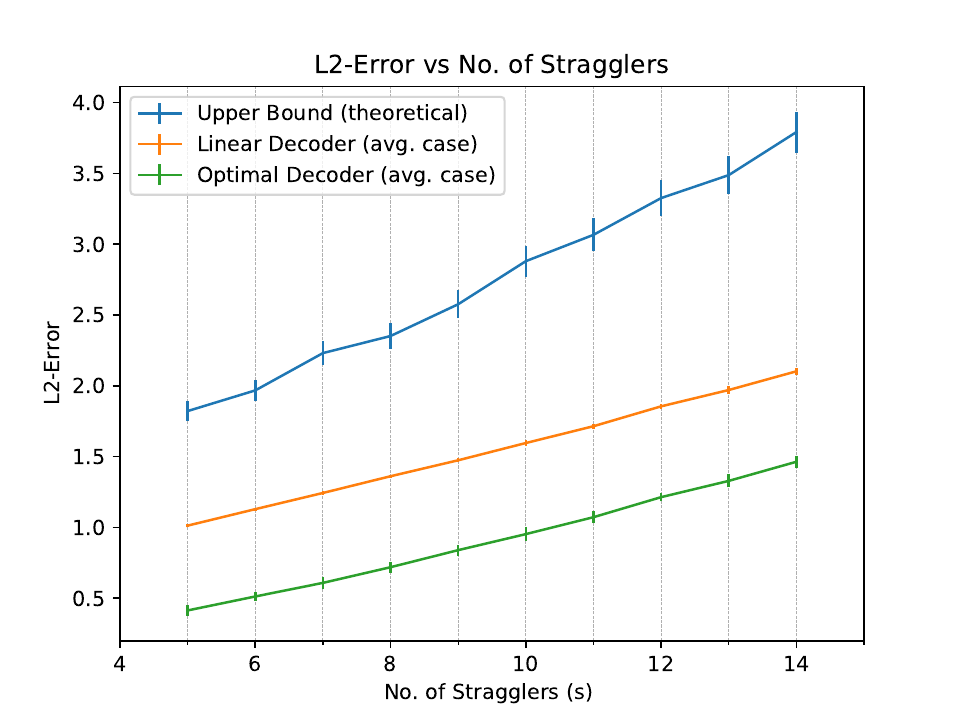}
		\caption{$n=30$, $d=5$}
		\label{fig:l2b}
	\end{subfigure}
	\begin{subfigure}[b]{0.3\textwidth}
		\includegraphics[width=\textwidth]{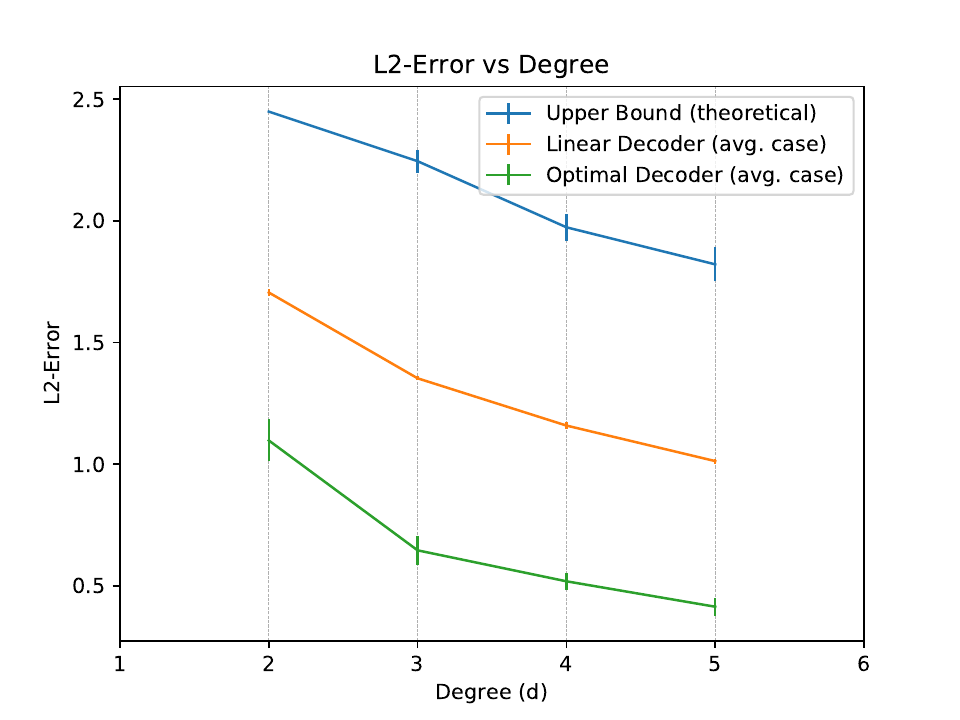}
		\caption{$n=30$, $s=5$}
		\label{fig:l2c}
	\end{subfigure}
	\begin{subfigure}[b]{0.3\textwidth}
		\includegraphics[width=\textwidth]{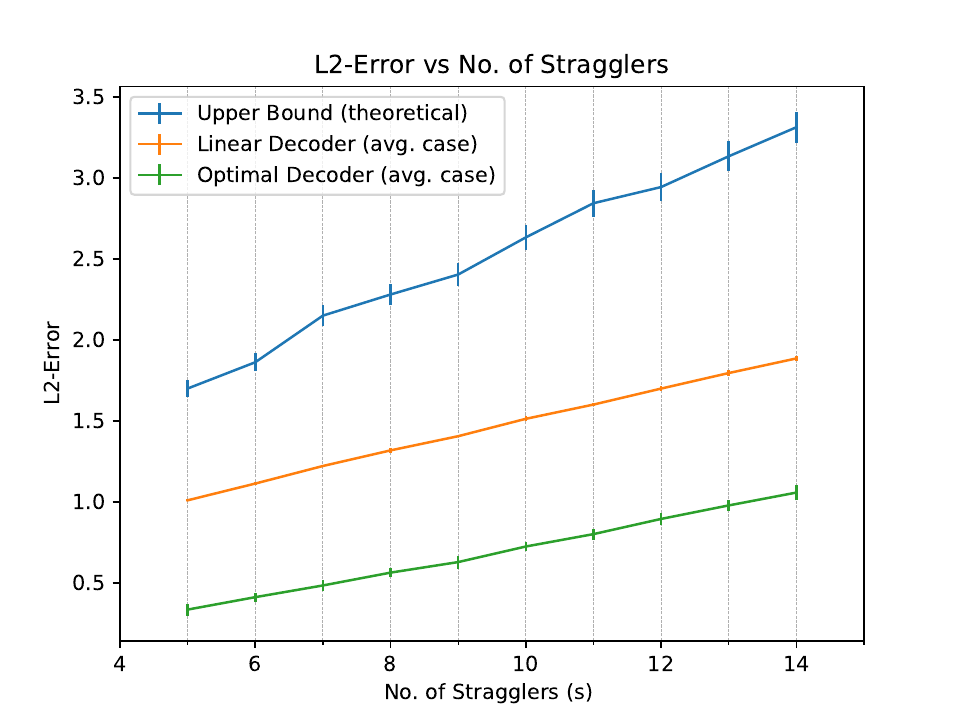}
		\caption{$n=50$, $d=5$}
		\label{fig:l2d}
	\end{subfigure}
	\begin{subfigure}[b]{0.3\textwidth}
		\includegraphics[width=\textwidth]{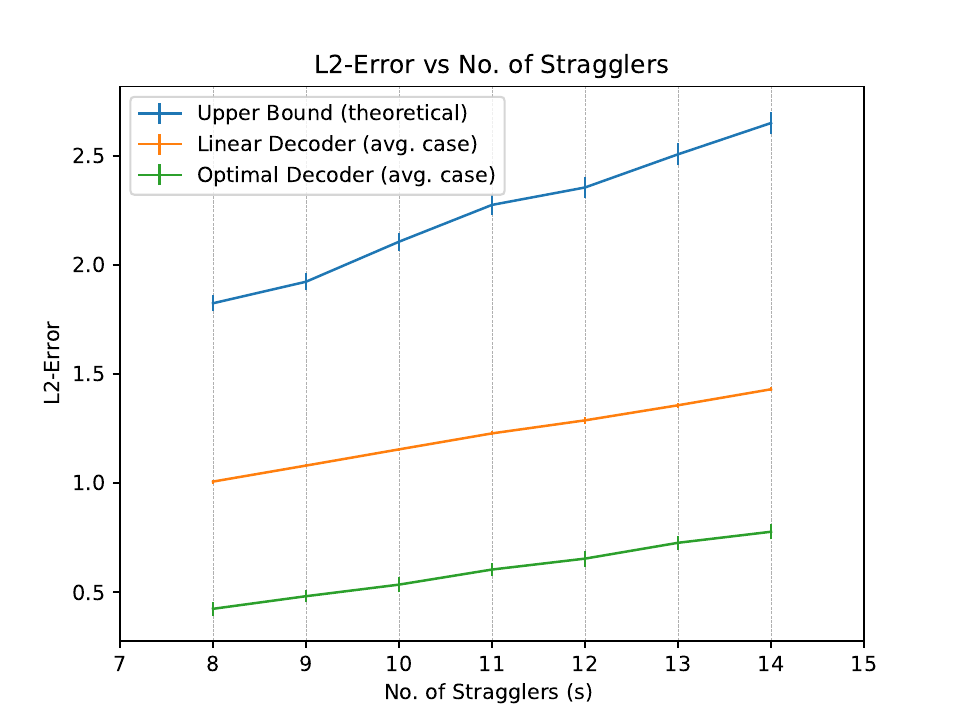}
		\caption{$n=50$, $d=8$}
		\label{fig:l2e}
	\end{subfigure}
	\begin{subfigure}[b]{0.3\textwidth}
		\includegraphics[width=\textwidth]{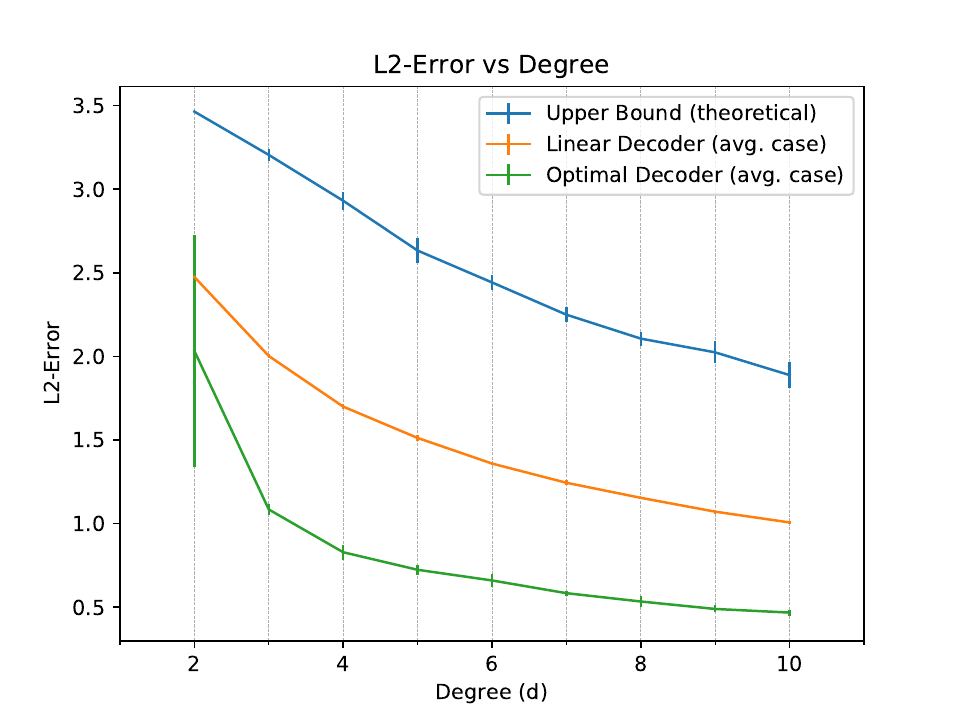}
		\caption{$n=50$, $s=5$}
		\label{fig:l2f}
	\end{subfigure}
	\caption{$\ell_2$-error for recovery of~$\1$ using normalized adjacency matrices of random $d$-regular graphs.}
	\label{fig:l2}
\end{figure*}
In this section, we present experimental results of our proposed approximate gradient coding scheme (Sec.~\ref{section:approximate}).

\subsection{$\ell_2$-Error} 
We measured the performance of our approximate coding schemes in terms of the $\ell_2$-error for the recovery of~$\1$. We chose the normalized adjacency matrix of a random $d$-regular graph on $n$ vertices as the matrix~$\boldB$. We randomly chose~$n-s$ rows of~$\boldB$ to be the surviving workers in any particular iteration, where~$s$ is the number of stragglers. For the decoding vector~$a(\cK)$, we chose the vector in~\eqref{equation:uK}, (called the \textit{linear decoder}), as well as the optimal least squares solution (called the \textit{optimal decoder}):
\begin{equation}
a(\cK) = \min_{\bolda\in\bR^{n-s}} \lVert \bolda \boldB(\cK,:) - \1\rVert_2,
\end{equation}
where~$\boldB(\cK:)$ is the submatrix of~$\boldB$ which consists of the rows that are indexed by~$\cK$.
Note that even though we have no additional theoretical guarantees for the optimal decoder, it is always possible to compute it in cubic time, e.g., by the singular value decomposition of~$\boldB(\cK,:)$.

Figure~\ref{fig:l2} presents the results using graphs on $n=30, 50$ vertices, and various values of~$s$ and~$d$. The results shown are averaged over multiple samples of~$\cK$ and multiple draws of the matrix~$\boldB$. Figures~\ref{fig:l2a}, \ref{fig:l2b}, \ref{fig:l2d}, and \ref{fig:l2e} show the $\ell_2$-error versus number of stragglers~$s$. As the number of stragglers increases, the recovery gets worse for a fixed degree~$d$.

Figures \ref{fig:l2c} and \ref{fig:l2f} show the $\ell_2$-error versus the degree~$d$. As~$d$ increases, the recovery error gets better for a fixed number of stragglers~$s$. Also, as expected, in all cases, the optimal decoder does better than the linear decoder in terms of the $\ell_2$-error. Interestingly, we can also observe that on average both the linear decoder and the optimal decoder are better than the theoretical upper bound in our paper. One could even think of exploiting this empirically by randomizing the assignment of the rows of~$\boldB$ to the different workers in every iteration.

\subsection{Generalization Error} 
In this section, our approximate gradient coding (AGC) scheme is compared to other approaches. We compare against gradient coding from~\cite{TandonLDK17} (GC), as well as the trivial scheme (IS), where the data is divided equally among all workers, but the master only uses the first~$n-s$ gradients.

We measured the performance of our coding schemes in terms of the \text{area under the curve} (AUC) on a validation set for a logistic regression problem, on a real dataset. The dataset we used was the Amazon Employee dataset from Kaggle. We used~$26,200$ training samples, and a model dimension of~$241,915$ (after one-shot encoding with interaction terms), and used gradient descent to train the logistic regression. For GC we used a constant learning rate, chosen using cross-validation. For AGC~and IS we used a learning rate of~$c_1/(r + c_2)$, which is typical for SGD, where~$c_1$ and~$c_2$ were also chosen via cross-validation.

All our methods were implemented in python using MPI4py (similar to \cite{TandonLDK17}). We ran our experiments using \texttt{t2.micro} worker instance types on Amazon EC2 and a \texttt{c3.8xlarge} master instance type. The results for~$n = 30, 50$ are given in Fig.~\ref{fig:auc1} and Fig.~\ref{fig:auc2}, in which AGC~corresponds to our approximation schemes with the optimal decoder, whereas AGC~(Linear), termed AGCL~is our full proposed approximation scheme. 

We observe that both these approaches are only slightly worse than GC, which utilizes the full gradient, and are quite better than the IS approach. Compared to each other, AGC~and AGCL~seem equivalent, however AGC~was marginally better. That being said, AGCL~can be faster since computing the linear decoder only requires~$O(n)$ time, in contrast to~$O(n^3)$ time for the optimal decoder. 

\begin{figure}[ht]
	\centering
	\includegraphics[width=0.4\textwidth]{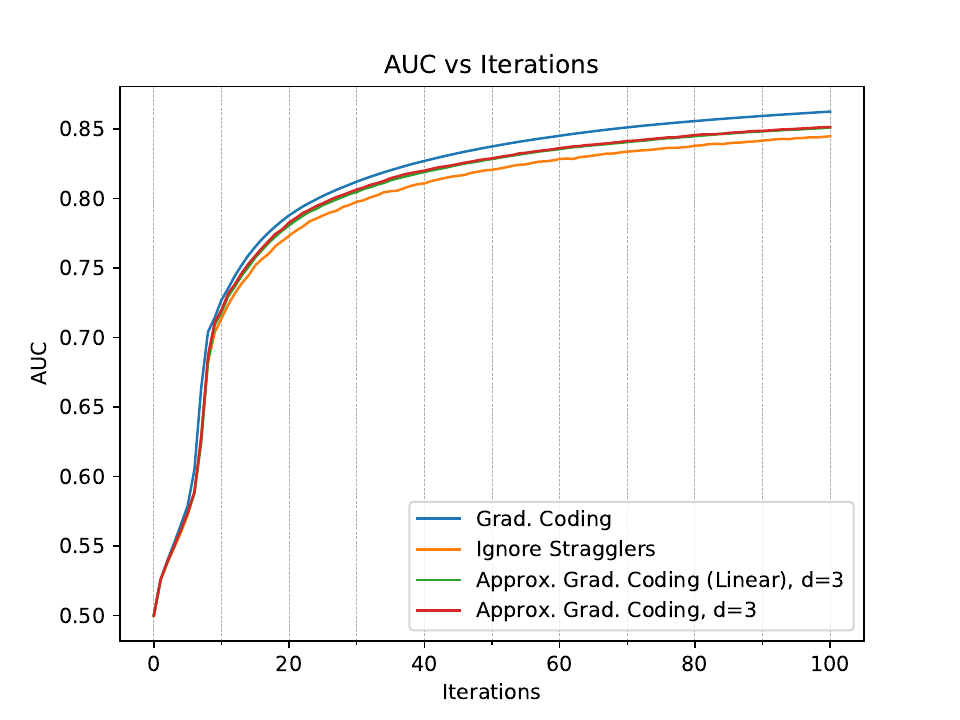}
	\caption{Generalization error versus number of iterations using $n=30$ \texttt{t2.micro} worker instances on EC2, with~$d=3$, and~$s=5$. Note that in case of \GC~\cite{TandonLDK17}, the computational overhead here is $\times6$ times (instead of $\times3$ in our approach).}
	\label{fig:auc1}
\end{figure}

\begin{figure}[ht]
	\centering
	\includegraphics[width=0.4\textwidth]{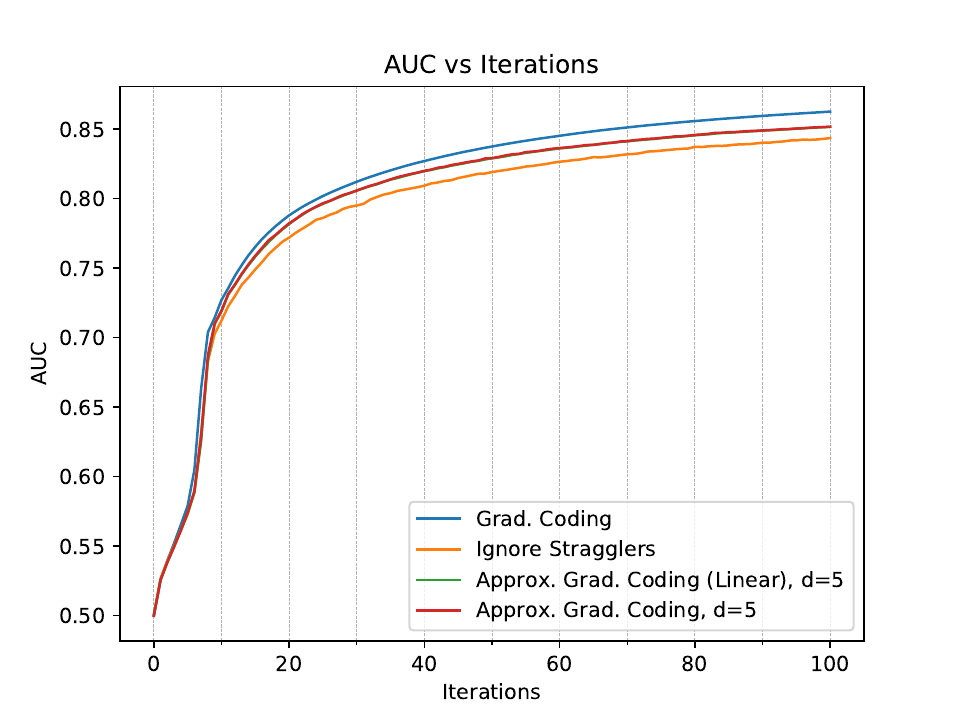}
	
	\caption{Generalization error versus number of iterations using~$n=50$ \texttt{t2.micro} worker instances on EC2, with~$d=5$, and~$s=10$. Note that in case of \GC~\cite{TandonLDK17}, the computational overhead here is $\times11$ times (instead of $\times5$ in our approach).}
	\label{fig:auc2}
\end{figure}

\section*{Acknowledgments}
This research has been supported by NSF Grants CCF
1422549, 1618689, DMS 1723052, ARO YIP W911NF-
14-1-0258 and research gifts by Google, Western Digital
and NVIDIA. The work of Rashish Tandon was done while
he was at UT Austin, prior to joining apple. The work of Itzhak Tamo and Netanel Raviv was supported in part ISF Grant 1030/15 and NSF-BSF Grant 2015814. The work of Netanel Raviv was supported in part by the postdoctoral fellowship of the Center for the Mathematics of Information (CMI) in the California Institute of Technology, and in part by the Lester-Deutsch postdoctoral fellowship. The authors express their gratitude to Prof.~Roi Livni for his valuable input.

\bibliographystyle{IEEEtranS}
\bibliography{bib}

\ifSUPPLEMENTARY
\appendices

\section{}\label{section:AK_complex}
In this section, efficient algorithms for encoding (i.e., computing the matrix~$\boldB$) and decoding (i.e., computing the vector~$a(\cK)$ given a set~$\cK$ of non-stragglers) are given for the scheme in Section~\ref{section:ComplexCyclic}.

Since the matrix~$\boldB$ is circulant, it suffices to compute only its leftmost column. Further, the leftmost column~$\boldc^\top_1$ of~$\boldB$ is a codeword in an~$[n,n-s]$ Reed-Solomon code whose evaluation points are all roots of unity of order~$n$, denoted~$\{\alpha_i\}_{i=0}^{n-1}$. Hence, to find~$\boldc_1$, one can define the polynomial~$m(x)\triangleq \prod_{j=s+1}^{n-1}(x-\alpha_j)$ and evaluate it over~$\alpha_0,\alpha_1,\ldots,\alpha_s$, which is possible in~$O(s(n-s))$ operations. This compares favorably with the respective~$O(n^2\log^2(n))$ of~\cite{ShortDot} (Sec.~5.1.1) for any~$s$.

As for decoding, for a given set~$\cK$ of~$n-s$ non-stragglers we present an algorithm which computes~$a(\cK)$ in~$O(s\log^2s+n\log n)$ operations over~$\bC$. This outperforms the respective~$O((n-s)\log^2(n-s))$ in~\cite{ShortDot} (Sec.~5.2.1) whenever~$s=o(n)$, and the respective~$O((n-s)^2)$ of~\cite{Wael} for any~$s<\delta n$ with~$\delta<1$. The central tool in this section is the well-known \textit{Generalized Reed-Solomon} (GRS) codes. A code~$\cC\subseteq \bC^n$ is called an~$[n,s]$ GRS code~if
\begin{align*}
	\cC=\left\{ (\ell_1f(\alpha_1),\ell_2f(\alpha_2),\ldots,\ell_nf(\alpha_n))~:~f\in\bC[x]^{<s} \right\},
\end{align*}
where~$\{\alpha_i\}_{i=1}^n$ are pairwise distinct evaluation points, and~$\{\ell_i\}_{i=1}^n\subseteq \bC$ are nonzero \textit{column multipliers}. It is readily verified that any RS code (Section~\ref{section:notions}) is a GRS code.

\begin{lemma}[{\cite{Ronny}, Prop.~5.2}]\label{lemma:GRS}
	If~$\cC\subseteq \bF^n$ is a GRS code then its dual~$\cC^\bot$ is a GRS code with identical evaluation points.
\end{lemma}

Let~$\cC$ be the code from Subsection~\ref{section:ComplexCyclic}, and notice that it is a GRS code whose column multipliers are all equal to~$1$. By Lemma~\ref{lemma:GRS}, it follows that the generator matrix of~$\boldC^\bot$ is~$\boldV\cdot \boldD$, where
\begin{align*}
\boldV=
\begin{pmatrix}
1 & 1 & \ldots & 1 \\
\alpha_1 & \alpha_2 & \ldots & \alpha_n \\
\alpha_1^2 & \alpha_2^2 & \ldots & \alpha_n^2 \\
\vdots & \vdots & \vdots & \vdots \\
\alpha_1^{s-1} & \alpha_2^{s-1} & \ldots & \alpha_n^{s-1} \\
\end{pmatrix}
\end{align*} 
and~$\boldD=\diag(\ell_1,\ldots,\ell_n)$ is a diagonal matrix which contains the nonzero column multipliers of~$\boldC^\bot$. In Algorithm~\ref{algorithm:AK_complex}, for a subset~$\cK\subseteq [n]$ let~$\boldV_\cK$ be the matrix of \textit{columns} of~$\boldV$ that are indexed by~$\cK$, and let~$\boldD_\cK=\diag(\ell_i)_{i\in \cK}$. 
In addition, for a vector~$\boldx\in\bC^n$, let~$\boldx_\cK$ be the vector which results from deleting the entries of~$\boldx$ that are not in~$\cK$.

\begin{algorithm}[tb]
   \caption{Computing the decoding vector~$a(\cK)$ for the scheme in Subsection~\ref{section:ComplexCyclic}.}\label{algorithm:AK_complex}
   {\bfseries Data:} Any vector~$\boldx'\in \bC^n$ such that~$\boldx'\boldB=\1$.\\
   {\bfseries Input:} A set~$\cK\subseteq[n]$ of~$n-s$ non-stragglers.\\
   {\bfseries Output:} A vector~$a(\cK)$ such that~$\support{a(\cK)}\subseteq \cK$ and~$a(\cK)\boldB=\1$.\\
   Find~$\boldf\in \bC^s$ such that~$\boldf\boldV_{\cK^c}=-\boldx'_{\cK^c}\cdot \boldD_{\cK^c}^{-1}$.\\
   Let~$\boldy\triangleq \boldf\cdot \boldV\boldD$.\\
   {\bfseries return}~$a(\cK)=\boldy+\boldx'$

\end{algorithm}


\begin{proof}[Correctness of Algorithm~\ref{algorithm:AK_complex}]
	Since~$\boldy\in \cC^\bot$, it follows that~$\boldy\boldB=0$, and hence~$a(\cK)\boldB=(\boldy+\boldx')\boldB=\1$. In addition, \blue{since $\boldy=\boldf\cdot \boldV\boldD$, it follows that~$\boldy_{\cK^c}=(\boldf\cdot \boldV\boldD)_{\cK^c}=\boldf\cdot \boldV_{\cK^c} \boldD_{\cK^c}=-\boldx'_{\cK^c}$, and thus $y_i=-x'_i$ for every~$i\notin K$. Therefore, it follows that $a(\cK)_i=y_i+x'_i=0$ for every~$i\notin \cK$, which implies that $\support{a(\cK)}\subseteq \cK$.}
\end{proof}

\begin{proof}[Complexity of Algorithm~\ref{algorithm:AK_complex}]
	Notice that --
	\begin{enumerate}
		\item Since~$\boldD_{\cK^c}$ is diagonal, computing its inverse requires~$s$ inverse operations in~$\bC$.
		\item Solving the equation $\boldf\boldV_{\cK^c}=-\boldx'_{\cK^c}\cdot \boldD_{\cK^c}^{-1}$ amounts to an interpolation problem, i.e., finding a degree (at most)~$s-1$ polynomial which passes through~$s$ given points. This is possible in~$O(s\log^2s)$ operations by~\cite{Kung}.
		\item Given~$\boldf$, computing the product~$\boldf\boldV$ reduces to evaluation of a degree (at most)~$n$ polynomial on all roots of unity of order~$n$. This is possible in~$O(n\log n)$ operations by utilizing the famous Fast Fourier Transform (FFT)~\cite{FFT}.
	\end{enumerate}
Hence, the total complexity of Algorithm~\ref{algorithm:AK_complex} is~$O(s\log^2s+n\log n)$.
\end{proof}

The pre-computation of~$\boldx'$ may be done by finding~$\boldx''\in\bC^{n-s}$ such that~$\boldx''\boldB'=\1$, where~$\boldB'$ is the upper-left $(n-s)\times(n-s)$ submatrix of~$\boldB$, and padding~$\boldx''$ with~$s$ zeros. Since~$\boldB'$ is a lower-triangular matrix in which the support of every column is of size at most~$s+1$, the equation~$\boldx''\boldB'=\1$ can be solved by a simple~$O(s^2+s(n-s))$ back-substitution algorithm. The computation of~$\boldx'$ can be done at the encoding phase (described above), and at a comparable complexity.

\section{}\label{section:AK_real}
As in the complex number case (Subsection~\ref{section:ComplexCyclic} and Appendix~\ref{section:AK_complex}), an algorithm for computing the matrix~$\boldB$ and the vector~$a(\cK)$ for the scheme in Subsection~\ref{section:RealCyclic} is given. 
In either of the cases of Construction~\ref{construction:realBCH}, the resulting code~$\cC$ consists of all codewords (seen as coefficients of polynomials in~$\bR[x]^{<n}$) with~$s$ mutual roots. That is, the code can be described as the right kernel of 
\begin{align*}
	\boldV\triangleq\begin{pmatrix}
		1 & \alpha_1 & \alpha_1^2 & \ldots & \alpha_1^{n-1}\\
		1 & \alpha_2 & \alpha_2^2 & \ldots & \alpha_2^{n-1}\\
		\vdots & \vdots & \vdots & \ddots & \vdots \\
		1 & \alpha_s & \alpha_s^2 & \ldots & \alpha_s^{n-1}	\\
	\end{pmatrix},
\end{align*}
where~$\{ \alpha_i \}_{i=1}^s$ are the roots of the code. 

As in Appendix~\ref{section:AK_complex},
to compute the matrix~$\boldB$ it suffices to find the codeword~$\boldc_1$, which in this case is a lowest weight codeword in a BCH code. It is readily verified that~$\boldc_1$ may be given by the coefficients of the \textit{generator polynomial} $g(x)\triangleq \prod_{i=1}^{s}(x-\alpha_i)$ (denoted by~$p_1$ and~$p_2$ in Eq.~\eqref{equation:generatorPolys}) . Finding these coefficients is possible by evaluating~$g(x)$ in~$s+1$ arbitrary and pairwise distinct roots of unity of order~$n$, and solving an interpolation problem in~$O(s\log^2s)$ by~\cite{Kung}. However, evaluating~$g(x)$ at~$s+1$ points requires~$O(n\log n)$ operations using the FFT algorithm. Hence, the overall complexity of computing~$\boldB$ is~$O(\min\{ s\log^2s,n\log n \})$, an improvement over~\cite{ShortDot} whenever~$s=o(n)$.

In Algorithm~\ref{algorithm:AK_real}, for a given~$\cK\subseteq[n]$, let~$\boldV_\cK$ be the matrix of columns of~$\boldV$ that are indexed by~$\cK$. The complexity of this algorithm outperforms~\cite{ShortDot} whenever~$s=o(\log^2 n)$, and outperforms~\cite{Wael} whenever~$s=o(n^{2/3})$.

\begin{algorithm}[tb]
   \caption{Computing the decoding vector~$a(\cK)$ for the scheme in Subsection~\ref{section:RealCyclic}.}
   \label{algorithm:AK_real}
   {\bfseries Data:} Any vector~$\boldx'\in \bR^n$ such that $\boldx'\boldB=\1$.\\
   {\bfseries Input:} A set~$\cK\subseteq[n]$ of~$n-s$ non-stragglers.\\
   {\bfseries Output:} A vector~$a(\cK)$ such that~$\support{a(\cK)}\subseteq \cK$ and~$a(\cK)\boldB=\1$.\\
   Compute~$\boldV_{\cK^c}^{-1}$.\\
   Let~$\boldf\triangleq -\boldx'_{\cK^c}\cdot \boldV_{\cK^c}^{-1}$.\\
	Let~$\boldy\triangleq \boldf\cdot \boldV$.\\
   {\bfseries return}~$a(\cK)=\boldy+\boldx'$.
\end{algorithm}

\begin{proof}[Correctness of Algorithm~\ref{algorithm:AK_real}]
	First, note that~$\boldV_{\cK^c}$ is an invertible matrix, since~$\cC^\bot$ is an MDS code by Lemma~\ref{lemma:MDSproperties} and Lemma~\ref{lemma:BCHareCyclicMDS}. Second, since~$\boldy$ is in the left image of~$\boldV$ it follows that~$\boldy\in \cC^\bot$, and further, $y_i=-x'_i$ for every~$i\in \cK^c$. Hence, $\support{a(\cK)}\subseteq \cK$ and~$a(\cK)\boldB=\1$.
\end{proof}

\begin{proof}[Complexity of Algorithm~\ref{algorithm:AK_real}]
	As in the complexity analysis of Algorithm~\ref{algorithm:AK_complex}, the complexity is clearly~$O(\gamma_s+s(n-s))$, where~$\gamma_s$ is the complexity of inverting a generalized Vandermonde matrix. While explicit formulas for inverting a generalized Vandermonde matrix were discussed in~\cite{GeneralizedVandermonde}, for simplicity, one may employ the trivial~$O(s^3)$ algorithm.
\end{proof}

\section{Bandwidth reduction in Subsection~\ref{section:ComplexCyclic}}\label{section:BandwidthReduction}
In step~$r\in[t]$ of Algorithm~\ref{algorithm:distributedSGD}, each node~$W_j$ transmits ~$\frac{1}{n}\sum_{i\in\support{\boldb_j}}b_{j,i}\cdot \nabla L_{\cS_i}(\boldw^{(r)})=\boldb_j\cdot \boldN(\boldw^{(r)})$ back to the master node. In the scheme which is suggested in Subsection~\ref{section:ComplexCyclic}, this step requires transmitting a vector in~$\bC^p$, where~$p$ is the number of entries of~$\boldw^{(r)}$. Since the gradient vectors~$\nabla L_{\cS_i}(\boldw^{(r)})$ are real, the resulting bandwidth is larger than the actual amount of information by a factor of~$2$. In this section it is shown that the optimal bandwidth can be attained by a simple manipulation of the gradient vectors. 

For~$p'\triangleq\ceil{\frac{p}{2}}$ and any~$r\in[t]$, denote the columns of the matrix~$\boldN(\boldw^{(r)})\in\bR ^{n\times p}$ by~$\boldN_{r,1},\ldots,\boldN_{r,p}$, and let~$\boldN'_r\in\bC^{n\times p'}$ be the matrix whose columns are\footnote{The rightmost column of~$\boldN'_r$ is~$\boldN_{r,p-1}+i\boldN_{r,p}$ if~$p$ is even and~$\boldN_{r,p}$ otherwise.}~$\boldN_{r,1}+i\boldN_{r,2},\boldN_{r,3}+i\boldN_{r,4},\ldots~$. To obtain optimal bandwidth, replace the transmission~$\boldb_j\cdot \boldN(\boldw^{(r)})\in\bC^{p}$ of~$W_j$ at iteration~$r$ of Algorithm~\ref{algorithm:distributedSGD} by~$\boldb_j\cdot \boldN'_r\in\bC^{p'}$. In addition, replace the definition of~$\boldw^{(r+1)}$ by~$\boldw^{(r)}-\eta e(\boldv_r)$, for~$e:\bC^{p'}\to \bR^p$ that is defined as
\begin{align*}
e(\boldv)\triangleq (\Re(v_1),\Im(v_1),\Re(v_2),\Im(v_2),\ldots),
\end{align*}
where~$\Re(y)$ and~$\Im(y)$ are the real and imaginary parts of a given~$y\in\bC$, respectively.

Since~$a$ and~$\boldB$ of Subsection~\ref{section:ComplexCyclic} satisfy the~EC condition, it follows that~$a(\cK)\boldB\boldN'_r=\1 \cdot \boldN'_r$ for every set~$\cK$ of~$n-s$ non-stragglers and every~$r\in[t]$. Further, we have that
\begin{align*}
e(\1 \boldN'_r)&=e\left( \1  (\boldN_{r,1}+i\boldN_{r,2}), \1  (\boldN_{r,3}+i\boldN_{r,4}),\ldots \right)\\
&=e\left( \1  \boldN_{r,1}+i\1 \boldN_{r,2}), \1  \boldN_{r,3}+i\1 \boldN_{r,4}),\ldots \right)\\
&=\left(\1\boldN_{r,1},\1\boldN_{r,2},\1\boldN_{r,3}\ldots \right)=\1\boldN(\boldw^{(r)}),
\end{align*}
and hence the correctness of Algorithm~\ref{algorithm:distributedSGD} under the scheme from Subsection~\ref{section:ComplexCyclic} is preserved. Under this framework, the transmission from~$W_j$ to the master node contains~$p'$ complex numbers (i.e.,~$2p'$ real numbers) rather than~$p$ complex numbers (i.e.,~$2p$ real numbers), and it is clearly optimal.

Note that this improvement does not diminish the contribution of the scheme in Subsection~\ref{section:RealCyclic}. In the current section, to compute the transmission~$\boldb_j \boldN'_r$, any worker node~$W_j$ performs~$dp'$ multiplication operations over~$\bC$, i.e.,~$4dp'\approx 2dp$ multiplication operations over~$\bR$. In contrast, the scheme in Subsection~\ref{section:RealCyclic} requires any~$W_j$ to perform only~$dp$ multiplications over~$\bR$.

\color{black}
\section{}\label{appendix:convergenceProof}
We operate under the convention that if all servers are stragglers (i.e., $\cK=\varnothing$), then the algorithm outputs the vector~$0$ as the approximation of the gradient. To analyze the trivial algorithm under this convention, define the random variable
\begin{align*}
	Y=
	\begin{cases}
		\tfrac{n}{|\cK|}\cdot\1_\cK &\mbox{if }\cK\ne \varnothing\\
		0 &\mbox{else,}	
	\end{cases}
\end{align*}
and notice that
\begin{align*}
	\norm{\1-Y}^2=
	\begin{cases}
	\frac{ns}{n-s} &\mbox{if }\cK\ne \varnothing\mbox{ (similar to~\eqref{equation:trivialApproximationScheme}),}\\
	n&\mbox{else}.
	\end{cases}
\end{align*}
Therefore, it follows that
\begin{align}\label{equation:convergenceBoundComp}
	\tfrac{1}{n}\cdot \bE[\norm{\1-Y}^2]&=(1-q)^n +\sum_{\cK\in\cP(n)\setminus\{\varnothing\}}q^{|\cK|}(1-q)^{n-|\cK|}\cdot \tfrac{n-|\cK|}{|\cK|}\nonumber\\
	&=(1-q)^n+\sum_{i=1}^{n-1}\binom{n}{i}q^i(1-q)^{n-i}\cdot\tfrac{n-i}{i}\nonumber\\
	&=(1-q)^n+\sum_{i=1}^{n-1}\binom{n}{i+1}q^i(1-q)^{n-i}\cdot\tfrac{i+1}{i}\nonumber\\
	&\le (1-q)^n+2\sum_{i=1}^{n-1}\binom{n}{i+1}q^i(1-q)^{n-i}\nonumber\\
	&= (1-q)^n+2\sum_{i=2}^{n}\binom{n}{i}q^{i-1}(1-q)^{n-(i-1)}\nonumber\\
	&\le (1-q)^n+\tfrac{2(1-q)}{q}\sum_{i=0}^{n}\binom{n}{i}q^{i}(1-q)^{n-i}=(1-q)^n+\tfrac{2(1-q)}{q}.
\end{align}

Hence, in the trivial algorithm we have
\begin{align*}
	\bE\norm{\nabla L_\cS(\boldw)-\boldv(\boldw)}^2&\le \specnorm{\boldN(\boldw)}^2\cdot \bE_{\cK}\norm{\1-Y}^2\le \specnorm{\boldN(\boldw)}^2\cdot n\cdot\left( (1-q)^n+\tfrac{2(1-q)}{q} \right).
\end{align*}
To conduct a similar analysis for the scheme in Section~\ref{section:approximate} we observe that
\begin{align*}
\bE\norm{\nabla L_\cS(\boldw)-\boldv(\boldw)}^2&\le \specnorm{\boldN(\boldw)}^2\cdot \bE_{\cK}\norm{\1-(\1+\boldu_\cK)\boldB}^2.
\end{align*}
In the spirit of the definition of~$Y$ above, we define a random variable
\begin{align*}
	R=
	\begin{cases}
		\tfrac{\lambda^2}{d^2}\cdot \tfrac{ns}{n-s} &\mbox{if }\cK\ne \varnothing\\
		n & \mbox{else}.
	\end{cases}
\end{align*}
According to Lemma~\ref{lemma:ABproximity} it follows that $\bE_{\cK}\norm{\1-(\1+\boldu_\cK)\boldB}^2\le \bE[R]$, and by following computations similar to~\eqref{equation:convergenceBoundComp}, we obtain
\begin{align*}
	\tfrac{1}{n}\bE[R]\le (1-q)^n+\tfrac{\lambda^2}{d^2}\cdot\tfrac{2(1-q)}{q},
\end{align*}
which implies that
\begin{align*}
	\bE\norm{\nabla L_\cS(\boldw)-\boldv(\boldw)}^2\le\specnorm{\boldN(\boldw)}^2\cdot n\cdot\left(
(1-q)^n+\tfrac{\lambda^2}{d^2}\cdot\tfrac{2(1-q)}{q}	\right).
\end{align*}
Finally, we remark that the analysis in~\eqref{equation:convergenceBoundComp} does not preclude applying a more sophisticated approach to bound the error. Yet, it is readily verified that any such approach can also be used to bound~$\bE[R]$, which will always result in a better bound.
\fi
\end{document}